\documentclass[11pt,a4wide]{article}

\newcommand{\remove}[1]{}
\usepackage[utf8]{inputenc}
\usepackage[english]{babel}

\usepackage{tikz,slashbox}
\usepackage{xspace}

\usepackage{amsmath,amsfonts,amssymb,epsfig,color}
\usepackage{stmaryrd}
\usepackage{amssymb}
\usepackage{amsfonts}
\usepackage{amsmath}
\usepackage{hyperref}
\usepackage{amsthm}
\usepackage{amssymb,amstext,hyperref}
\usepackage{cite}
\usepackage{graphics}
\usepackage{amsmath,dsfont}
\usepackage{mathrsfs}
\usepackage{fancyhdr}
\usepackage{verbatim}
\usepackage{boxedminipage}
\usepackage{colortbl}
\usepackage{algorithm}
\usepackage[noend]{algorithmic}

\usepackage[T1]{fontenc}
\usepackage{tikz}
\usepackage{graphicx}

\newtheorem{theorem}{Theorem}
\newtheorem{lemma}{Lemma}
\newtheorem{corollary}{Corollary}

\newtheorem{prope}{Property}

\newtheorem*{T4}{Theorem~\ref{hardminor}}
\newtheorem*{T5}{Theorem~\ref{hardtminor}}

\newcommand{\smallo}{o}
\newcommand{\cupall}{\pmb{\pmb{\bigcup}}}

\newcommand{\Ostar}{\mathcal{O}^*}

\newcommand{\Acal}{\mathcal{A}}

\newcommand{\Ccal}{\mathcal{C}}

\newcommand{\Fcal}{\mathcal{F}}

\newcommand{\Ocal}{\mathcal{O}}
\newcommand{\Pcal}{\mathcal{P}}
\newcommand{\Qcal}{\mathcal{Q}}

\newcommand{\Scal}{\mathcal{S}}

\newcommand{\Nbb}{\mathbb{N}}

\newcommand{\ETH}{{\sf ETH}\xspace}

\newcommand{\polyn}{\cdot n^{\mathcal{O}(1)}}

\newcommand{\degs}[2]{{\sf deg}_{#1}(#2)}

\theoremstyle{plain}

\newcommand{\ig}[1]{\textcolor{red}{[Ig: #1]}}

\newcommand{\paraprobl}[5]
{
  \begin{flushleft}
    \fbox{
      \begin{minipage}{#5cm}
        \noindent {\textsc {#1}}\\
        {\bf Input:} #2\\
        {\bf Parameter:} #4\\
        {\bf Output:} #3
      \end{minipage}
    }
  \end{flushleft}
}

\newcommand{\sm}{\setminus}
\newcommand{\gm}{\setminus}

\newcommand{\tw}{{\sf{tw}}}
\newcommand{\pw}{{\sf{pw}}}

\newcommand{\es}{\emptyset}

\newcommand{\intv}[1]{\left [ #1 \right ]}

\newcommand{\pretp}{\preceq_{\sf tm}}
\newcommand{\prem}{\preceq_{\sf m}}

\newcommand{\besf}{{\sf besf}}

\newcommand{\paw}{{\sf paw}\xspace}

\newcommand{\chair}{{\sf chair}\xspace}

\newcommand{\banner}{{\sf banner}\xspace}

\newcommand{\dart}{{\sf dart}\xspace}

\newcommand{\gem}{{\sf gem}\xspace}
\newcommand{\claw}{{\sf claw}\xspace}

\newcommand{\leaf}{{\sf L}\xspace}

\usepackage{aecompl}
\usepackage{color}
\definecolor{linkcol}{rgb}{0,0,0.8}
\definecolor{citecol}{rgb}{0.65,0,0}
\definecolor{titlecol}{rgb}{0.65,0,0}

\hypersetup {bookmarksopen=true, pdftitle="Hitting (topological) minors on
  bounded treewidth graphs", pdfauthor="Julien Baste - Ignasi Sau - Dimitrios M. Thilikos",
  pdfsubject="Hitting (topological) minors on
  bounded treewidth graphs", 
  pdfmenubar=true, 
  pdfhighlight=/O, 
  colorlinks=true, 
  pdfpagemode=None, 
  pdfpagelayout=DoublePage, 
  pdffitwindow=true, 
  linkcolor=linkcol, 
  citecolor=citecol, 
  urlcolor=linkcol 
}

\newcommand{\vertex}[1]{\filldraw (#1) circle (2 pt);}

\newcommand{\bct}{\textsf{bct}}
\newcommand{\block}{\textsf{block}}
\newcommand{\cut}{\textsf{cut}}

\definecolor{gray0}{gray}{0.875}
\definecolor{gray1}{gray}{0.775}
\definecolor{gray2}{gray}{0.75}

\newcommand\cuparrow{%
  \mathrel{\ooalign{\hss$\cup$\hss\cr%
      \kern0.3ex\raise0.7ex\hbox{\scalebox{0.7}{$\downarrow$}}}}}

\newcommand\bigcuparrow{%
  \mathrel{\ooalign{\hss$\bigcup$\hss\cr%
      \kern0.55ex\raise0.7ex\hbox{\scalebox{0.7}{$\downarrow$}}}}}

\newcommand{\pbtm}{\textsc{$\Fcal$-TM-Deletion}\xspace}
\newcommand{\pbm}{\textsc{$\Fcal$-M-Deletion}\xspace}

\begin{document}

\title{\vspace{-.5cm}Hitting minors on bounded treewidth graphs.\\III. Lower bounds\thanks{Emails of authors: \texttt{julien.baste@uni-ulm.de}, \texttt{ignasi.sau@lirmm.fr}, \texttt{sedthilk@thilikos.info}.\vspace{.2cm}\newline \indent \!\!\!\!\!The results of this article are permanently available at \texttt{https://arxiv.org/abs/1704.07284}. Extended abstracts containing some of the results of this article appeared in the \emph{Proc. of the 12th International Symposium on Parameterized and Exact Computation  (\textbf{IPEC 2017})}~\cite{BasteST17}, in the \emph{Proc. of the 13th International Symposium on Parameterized and Exact Computation  (\textbf{IPEC 2018})}~\cite{BasteST18}, and in the \emph{Proc. of the 31st Annual ACM-SIAM Symposium on Discrete Algorithms (\textbf{SODA 2020})}~\cite{BasteSTSODA20}. Work supported by French projects DEMOGRAPH (ANR-16-CE40-0028) and ESIGMA (ANR-17-CE23-0010).\newline}}

\author{\bigskip Julien Baste\thanks{LIRMM, Université de Montpellier, Montpellier, France.}~\thanks{Sorbonne Université, Laboratoire d'Informatique de Paris 6, LIP6, Paris, France.}~\thanks{Institute of Optimization and Operations Research, Ulm University, Germany.} \and
  Ignasi Sau\thanks{LIRMM, Université de Montpellier, CNRS, Montpellier, France.}
  \and
  Dimitrios  M. Thilikos$^{\small \P}$}

\date{\vspace{-1cm}}

\maketitle

\begin{abstract}
   \noindent For a finite collection of graphs ${\cal F}$, the \textsc{$\Fcal$-M-Deletion} problem consists in, given a graph $G$ and an integer $k$, decide  whether there exists $S \subseteq V(G)$ with $|S| \leq k$ such that $G \setminus S$ does not contain any of the graphs in ${\cal F}$ as a minor. We are interested in the parameterized complexity of \textsc{$\Fcal$-M-Deletion} when the parameter is the treewidth of $G$, denoted by $\tw$.
   Our objective is to determine, for a fixed ${\cal F}$, the smallest function $f_{{\cal F}}$ such that \textsc{$\Fcal$-M-Deletion} can be solved in time $f_{{\cal F}}(\tw) \cdot n^{\Ocal(1)}$ on $n$-vertex graphs.
   We provide lower bounds under the \ETH on $f_{{\cal F}}$ for several collections ${\cal F}$. We first prove that for any $\Fcal$ containing connected graphs of size at least two, $f_{{\cal F}}(\tw)= 2^{\Omega(\tw)}$, even if the input graph $G$ is planar. Our main contribution consists of superexponential lower bounds for a number of collections $\Fcal$, inspired by a reduction of Bonnet et al.~[IPEC, 2017]. In particular, we prove that when ${\cal F}$ contains a single connected graph $H$
   that is either $P_5$ or is not a minor of the $\banner$ (that is, the graph consisting of a $C_4$ plus a pendent edge), then $f_{{\cal F}}(\tw)= 2^{\Omega(\tw \cdot \log \tw)}$. This is the third of a series of articles on this topic, and the results given here together with other ones allow us, in particular, to provide a tight dichotomy on the complexity of \textsc{$\{H\}$-M-Deletion}, in terms of $H$, when $H$ is connected.

\vspace{.3cm}

\noindent{\bf Keywords}: parameterized complexity; graph minors; treewidth; hitting minors; topological minors; dynamic programming; Exponential Time Hypothesis.
\vspace{.5cm}

\end{abstract}

\newpage

\section{Introduction}
\label{illusion}

Let ${\cal F}$ be a finite non-empty collection of non-empty graphs.  In the \textsc{$\Fcal$-M-Deletion} (resp. \textsc{$\Fcal$-TM-Deletion}) problem, we are given a graph $G$ and an integer $k$, and the objective is to decide whether there exists a set $S \subseteq V(G)$ with $|S| \leq k$ such that $G \setminus S$ does not contain any of the graphs in ${\cal F}$ as a minor (resp. topological minor). Instantiations of these problems correspond to several well-studied problems. For instance,  the cases ${\cal F}= \{K_2\}$, ${\cal F}= \{K_3\}$, and ${\cal F}= \{K_5,K_{3,3}\}$ of \textsc{$\Fcal$-M-Deletion} (or \textsc{$\Fcal$-TM-Deletion}) correspond to \textsc{Vertex Cover}, \textsc{Feedback Vertex Set}, and \textsc{Vertex Planarization}, respectively.



We are interested in the parameterized complexity of both problems when the parameter is the treewidth of the input graph. By Courcelle's theorem~\cite{Courcelle90}, \textsc{$\Fcal$-M-Deletion} \textsc{$\Fcal$-TM-Deletion} can be solved in time $f(\tw)\polyn$ on $n$-vertex graphs with treewidth at most $\tw$, where $f$ is some computable function.
Our objective is to determine, for a fixed collection ${\cal F}$, which is the {\sl smallest} such function $f$ that one can (asymptotically) hope for, subject to reasonable complexity assumptions.

This line of research has attracted some attention  in the parameterized complexity community during the last years. For instance, \textsc{Vertex Cover} is easily solvable in time $2^{\Ocal(\tw)}\polyn$, called \emph{single-exponential}, by standard dynamic-programming techniques, and no algorithm with running time $2^{o(\tw)}\polyn$ exists, unless the Exponential Time Hypothesis (\ETH)\footnote{The \ETH states that 3-\textsc{SAT} on $n$ variables cannot be solved in time $2^{o(n)}$; see~\cite{ImpagliazzoP01} for more details.} fails~\cite{ImpagliazzoP01}.

For \textsc{Feedback Vertex Set}, standard dynamic programming techniques give a running time of $2^{\Ocal(\tw \cdot \log \tw)}\polyn$, while the lower bound under the \ETH~\cite{ImpagliazzoP01} is again $2^{o(\tw)}\polyn$. This gap remained open for a while, until Cygan et al.~\cite{CyganNPPRW11} presented an optimal algorithm running in time $2^{\Ocal(\tw)}\polyn$, introducing  the celebrated \emph{Cut{\sl \&}Count}
technique. This article triggered several other techniques to obtain single-exponential algorithms for so-called \emph{connectivity problems} on graphs of bounded treewidth, mostly based on algebraic tools~\cite{BodlaenderCKN15,FominLPS16}.

Concerning \textsc{Vertex Planarization}, Jansen et al.~\cite{JansenLS14} presented an algorithm of time $2^{\Ocal(\tw \cdot \log \tw)}\polyn$ as a subroutine in an {\sf FPT}-algorithm parameterized by $k$. Marcin Pilipczuk~\cite{Pili15} proved that this running time is {\sl optimal} under the \ETH. This lower bound was acheived
by using the framework introduced by Lokshtanov et al.~\cite{permuclique,LokshtanovMS11} for obtaining superexponential lower bounds (namely, of the form $2^{\Omega(k \cdot \log k)}\polyn$, in particular for problems parameterized by treewidth), which has proved very successful in recent years~\cite{CyganNPPRW11,Pili15,BasteS15}. We also use this framework in the current article.

\newpage
\noindent
\textbf{Our results and techniques}. We present lower bounds under the \ETH for \textsc{$\Fcal$-M-Deletion} and \textsc{$\Fcal$-TM-Deletion} parameterized by treewidth, several of them being tight.
We first prove that for any connected\footnote{A \emph{connected} collection $\mathcal{F}$ is a collection containing only connected graphs of size at least two.} $\Fcal$, neither \textsc{$\Fcal$-M-Deletion} nor \textsc{$\Fcal$-TM-Deletion} can be solved in time $2^{o(\tw)}\polyn$, even if the input graph $G$ is planar (cf.~Theorem~\ref{temporis} and Corollary~\ref{patterns}). The main contribution of this article consists of superexponential lower bounds for a number of collections $\Fcal$, which we proceed to describe.
Let $\Ccal$ be the set of all connected graphs that contain a block (i.e., a biconnected component) with at least five edges,
let $\mathcal{Q}$ be the set containing $P_5$ and all connected graphs that are {\sl not} minors of the \banner (that is, the graph consisting of a $C_4$ plus a pendent edge),
and let $\Scal = \{K_{1,i} \mid i \geq 4 \}$. We prove that, assuming the \ETH,

\begin{itemize}
\item for every finite non-empty subset $\Fcal \subseteq \Ccal$, neither \textsc{$\Fcal$-M-Deletion} nor \textsc{$\Fcal$-TM-Deletion} can be solved in time $2^{\smallo(\tw \log \tw)}\polyn$ (cf.~Theorem~\ref{hardgene}),
\item for every $H \in \mathcal{Q}$, \textsc{$\{H\}$-M-Deletion} cannot be solved in time $2^{\smallo(\tw \log \tw)}\polyn$  (cf.~Theorem~\ref{hardminor}), and
\item  for every $H \in \mathcal{Q} \setminus \Scal$, \textsc{$\{H\}$-TM-Deletion} cannot be solved in time $2^{\smallo(\tw \log \tw)}\polyn$  (cf.~Theorem~\ref{hardtminor}).
\end{itemize}

\medskip

The general lower bound of  $2^{o(\tw)}\polyn$ for connected collections is based on a simple reduction from \textsc{(Planar) Vertex Cover}. The superexponential lower bounds, namely $2^{o(\tw \cdot \log \tw)}\polyn$, are  based on the ideas presented by Bonnet et al.~\cite{BonnetBKM-IPEC17} for generalized feedback vertex set problems. We provide two hardness results that apply to different families of collections $\Fcal$,  both  based on a general framework described in Section~\ref{supexpgencons}, consisting of a reduction from the {\sc $k\times k$ Permutation Independent Set} problem introduced by Lokshtanov et al.~\cite{permuclique}. Namely, we prove, in Theorem~\ref{hardgene} (applying to both the minor and topological minor versions), the lower bound when $\Fcal$ is any finite non-empty subset of all connected graphs that contain a block with at least five edges. We then prove, in Theorem~\ref{hardminor} (for minors) and Theorem~\ref{hardtminor} (for topological minors), the lower bound when $\Fcal$ contains a single graph $H$ that is either $P_5$ or is not a minor of the $\banner$, with the exception of $K_{1,i}$ mentioned above  for the topological minor version. The proofs of the latter two theorems are considerably longer, as we need to distinguish several cases according to certain properties of the graph $H$ (cf. Lemma~\ref{cycletwocut} up to Lemma~\ref{crickbistm}).

We would like to mention that in previous versions of this article (in particular, in the conference version presented in~\cite{BasteST17}), we presented another family of reductions inspired by a reduction of Pilipczuk~\cite{Pili15} for \textsc{Vertex Planarization}, that is, for $\Fcal = \{K_5, K_{3,3}\}$. Afterwards, we found a more general unifying reduction along the lines of Bonnet et al.~\cite{BonnetBKM-IPEC17}, which is the one we present here. This reduction generalizes the hardness results presented in~\cite{BasteST17} and in~\cite{BasteST18} (also inspired by~\cite{BonnetBKM-IPEC17}, and that can be seen as a weaker version of the current reduction), as well as the lower bound of Pilipczuk~\cite{Pili15}, which is a corollary of one of our hardness results, namely Theorem~\ref{hardgene}.

\medskip

\noindent\textbf{Results in other articles of the series and discussion}. In the first article of this series~\cite{monster1}, we show, among other results, that for every connected $\mathcal{F}$ containing at least one planar graph (resp. subcubic planar graph), \textsc{$\Fcal$-M-Deletion} (resp. \textsc{$\Fcal$-TM-Deletion}) can be solved in time
  $2^{\Ocal(\tw  \cdot \log \tw)}\polyn$. In the second article of this series~\cite{monster2}, we provide single-exponential algorithms for \textsc{$\{H \}$-M-Deletion} for all the graphs $H$ for which the superexponential lower bounds given in this article do not apply, namely those on the left of Figure~\ref{shifting}: $P_3$, $P_4$, $C_4$, the \claw, the \paw, the \chair (sometimes also called \emph{fork} in the literature), and the \banner. Note that the cases $H = P_2$~\cite{ImpagliazzoP01,CyganFKLMPPS15}, $H = P_3$~\cite{P3-cover,P3-cover-improved}, and $H = C_3$~\cite{CyganNPPRW11,BodlaenderCKN15} were already known (nevertheless, for completeness we provide in~\cite{monster2} a simple algorithm when $H = P_3$).
  In the fourth article of this series~\cite{BasteST20-SODA} (whose full version is~\cite{SODA-arXiv}), we present an algorithm for \textsc{$\Fcal$-M-Deletion} in time $\Ostar(2^{O(\tw \cdot \log \tw)})$ for {\sl any} collection $\Fcal$.

 \begin{figure}[h!]
  \begin{center}
    \includegraphics[width=.87\textwidth]{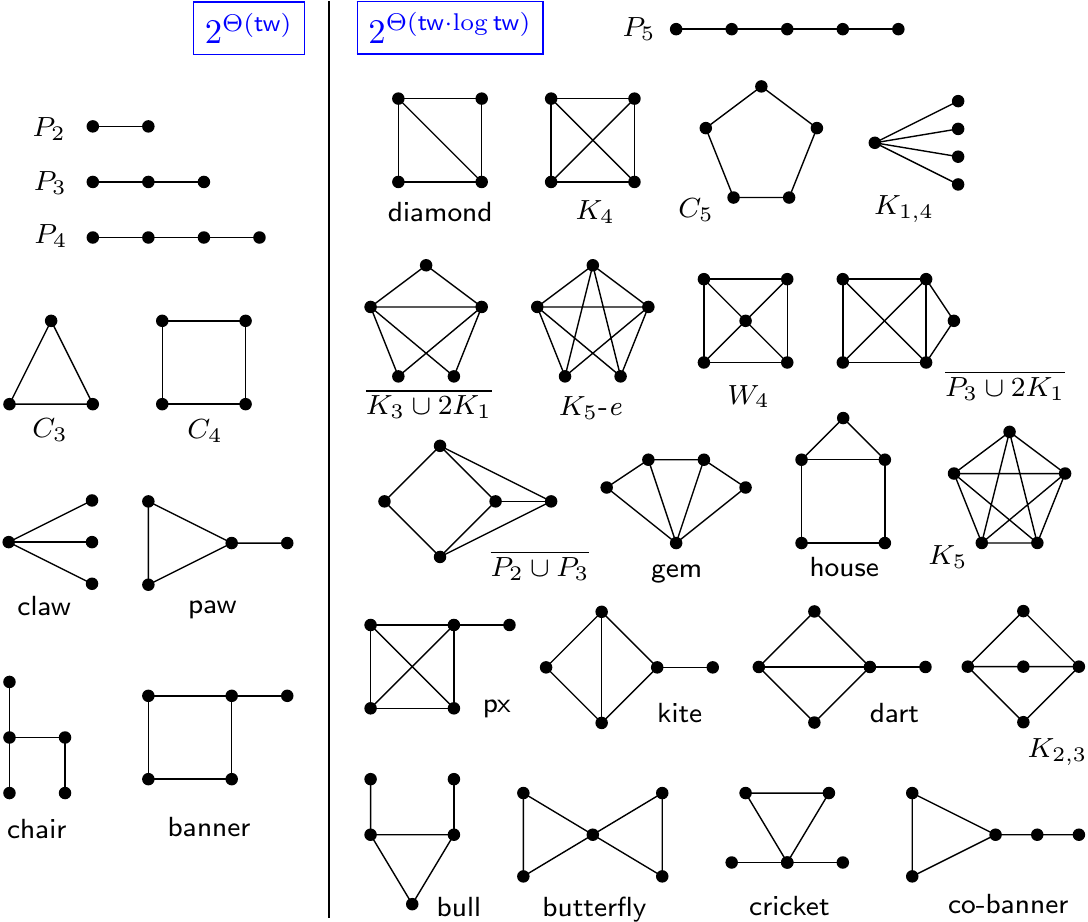}
  \end{center}\vspace{-.25cm}
  \caption{Classification of the complexity of \textsc{$\{H\}$-M-Deletion} for all connected simple graphs $H$ with $2 \leq |V(H)|\leq 5$: for the nine graphs on the left (resp. 21 graphs on the right, and all the larger ones), the problem is solvable in time $2^{\Theta(\tw)} \cdot n^{\Ocal(1)}$ (resp. $2^{\Theta(\tw \cdot \log \tw)} \cdot n^{\Ocal(1)}$). For \textsc{$\{H\}$-TM-Deletion}, $K_{1,4}$ should be on the left. This figure also appears in~\cite{monster2}.}
  \label{shifting}
  \vspace{.35cm}
\end{figure}

The  lower bounds presented in this article, together with the algorithms given in~\cite{monster2,SODA-arXiv,BasteST20-SODA}, cover all the cases of \textsc{$\Fcal$-M-Deletion}  where $\Fcal$ contains a single connected graph, as discussed in Section~\ref{secsupexp}. Namely,  we obtain the following theorem.

\begin{theorem}
Let $H$ be a connected  graph of size at least two. Then the \textsc{$\{H\}$-M-Deletion} problem can be solved in time
\begin{itemize}
    \item $2^{\Theta(\tw )}\polyn$, if $H$ is a minor of the \banner
    that is different from $P_5$, and
    \item $2^{\Theta(\tw \cdot \log \tw)}\polyn$, otherwise.
\end{itemize}
\end{theorem}

Note that the graphs $H$ described in the first item can be equivalently characterized as those that can be obtained by the \chair or the \banner by contracting edges. In the above statement, we use the $\Theta$-notation to indicate that these algorithms are {\sl optimal} under the \ETH. This dichotomy is depicted in Figure~\ref{shifting}, containing all  connected graphs $H$ with $2 \leq |V(H)| \leq 5$; note that if $|V(H)| \geq 6$, then $H$ is not a minor of the \banner, and therefore the second item above applies.
    See~\cite{monster2} for a discussion about the role played by the \banner in this dichotomy.

    The single-exponential algorithms given in~\cite{monster2} also apply to \textsc{$\{H \}$-TM-Deletion} for which, in addition, we provide
     a single-exponential algorithm when $H=K_{1,i}$ for every $i \geq 1$. By Theorem~\ref{hardminor}, a single exponential-algorithm for  \textsc{$\{K_{1,i}\}$-M-Deletion} when $i \geq 4$ is unlikely to exist.
  To the best of our knowledge, this is the first example of a collection $\Fcal$ for which the complexity of \textsc{$\Fcal$-M-Deletion} and \textsc{$\Fcal$-TM-Deletion} differ.

\medskip
\noindent
\textbf{Organization of the paper}. In Section~\ref{banality} we give some preliminaries.
In Section~\ref{struggle} we present the single-exponential lower bound for any connected $\Fcal$, and in Section~\ref{secsupexp} the superexponential lower bounds.
We conclude the article in Section~\ref{upheaval} with some open questions for further research.

%
%
%
%
%


\section{Preliminaries}
\label{banality}




In this section we provide some preliminaries to be used in the following sections.


\bigskip
\noindent
\textbf{Sets, integers, and functions.}
We denote by $\Nbb$ the set of every non-negative integer and we
set $\Nbb^+=\Nbb\setminus\{0\}$.
Given two integers $p$ and $q$, the set $\intv{p,q}$
refers to the set of every integer $r$ such that $p \leq r \leq q$.
Moreover, for each integer $p \geq 1$, we set $\Nbb_{\geq p} = \Bbb{N}\setminus\intv{0,p-1}$.
In the set $\intv{1,k}\times \intv{1,k}$, a \emph{row} is a set $\{i\} \times \intv{1,k}$ and a \emph{column} is a set $\intv{1,k} \times \{i\}$ for some $i\in \intv{1,k}$.

We use $\emptyset$ to denote the empty set and
$\varnothing$ to denote the empty function, i.e., the unique subset of $\emptyset\times\emptyset$.
Given a function $f:A\to B$ and a set $S$, we define $f|_{S}=\{(x,f(x))\mid x\in S\cap A\}$.
Moreover if $S \subseteq A$, we set $f(S) = \bigcup_{s \in S} \{f(s)\}$.
Given a set $S$, we denote by ${S \choose 2}$ the set containing
every subset of $S$ that has cardinality two. We also denote by $2^{S}$ the set of all the subsets of $S$.
If ${\cal S}$ is a collection of objects where the operation $\cup$ is defined, then we  denote $\cupall{\cal S}=\bigcup_{X\in{\cal S}}X$.

Let $p\in\Bbb{N}$ with $p\geq 2$, let $f:\Bbb{N}^{p}\rightarrow\Bbb{N}$, and let  $g:\Bbb{N}^{p-1}\rightarrow\Bbb{N}$.
We say that $f(x_{1},\ldots,x_{p})=\Ocal_{x_p}(g(x_{1},\ldots,x_{p-1}))$ if there is a function   $h:\Bbb{N}\rightarrow\Bbb{N}$ such that  \\ $f(x_{1},\ldots,x_{p})=\Ocal(h(x_p)\cdot g(x_{1},\ldots,x_{p-1}))$.

\bigskip
\noindent
\textbf{Graphs}.  All the graphs that we consider in this paper are undirected, finite, and without  loops or multiple edges.
We use standard graph-theoretic notation, and we refer the reader to~\cite{Die10} for any undefined terminology. Given a graph $G$, we denote by $V(G)$ the set of vertices of  $G$ and by $E(G)$ the set of the edges of $G$.
We call $|V(G)|$ {\em the size} of $G$. A graph is {\em the empty} graph if its size is zero.
We also denote by $\leaf(G)$ the set of the vertices of $G$ that have degree exactly
one in the case where $|V(G)| \geq 2$, and $\leaf(G) = V(G)$ if $|V(G)| = 1$.
If $G$ is a tree (i.e., a connected acyclic graph) then $\leaf(G)$ is the set of the {\em leaves}
of $G$.
A {\em vertex labeling} of $G$ is some injection $\rho: V(G)\to\Bbb{N}^{+}$. 
{Given a vertex $v \in V(G)$, we define the \emph{neighborhood} of
  $v$ as $N_G(v) = \{u \mid u \in V(G), \{u,v\} \in E(G)\}$ and the \emph{closed neighborhood}
  of $v$ as $N_G[v] = N_G(v) \cup \{v\}$.} If $X\subseteq V(G)$, then we write $N_{G}(X)=(\bigcup_{v\in X}N_{G}(v))\setminus X$.
The {\em degree} of a vertex $v$ in $G$ is defined as $\degs{G}{v}=|N_{G}(v)|$. A graph is called {\em subcubic}
if all its vertices have degree at most three.

A \emph{subgraph} $H = (V_H,E_H)$ of a graph $G$ is a graph such that $V_H \subseteq V(G)$ and $E_H \subseteq E(G) \cap {V_H \choose 2}$.
If $S \subseteq V(G)$, the subgraph of $G$ \emph{induced by} $S$, denoted $G[S]$, is the graph $(S, E(G) \cap {S \choose 2})$.
We also define $G \gm S$ to be the subgraph of $G$ induced by $V(G) \sm S$.
If $S \subseteq E(G)$, we denote by $G \gm S$ the graph $(V(G), E(G) \sm S)$.
%

If $s,t\in V(G)$ are two distinct vertices, an {\em $(s,t)$-path} of $G$
is any connected subgraph $P$ of $G$ with maximum degree two and where $s,t\in \leaf(P)$, and a \emph{path} is an $(s,t)$-path for some vertices $s$ and $t$.
We finally denote by ${\cal P}(G)$ the set of all paths of $G$.
{Given $P \in \Pcal(G)$, we say that $v \in V(P)$ is an \emph{internal vertex} of $P$ if $\degs{P}{v} = 2$.}
{Given an integer $i$ and a graph $G$, we say that $G$ is $i$-connected if for each $\{u,v\} \in {V(G) \choose 2}$,
  there exists a set $\Pcal' \subseteq \Pcal(G)$ of $(u,v)$-paths of $G$ such that
  $|\Pcal'| = i$ and
  for each $P_1,P_2 \in \Pcal'$ such that $P_1 \not = P_2$, $V(P_1) \cap V(P_2) = \{u,v\}$.}

We denote by $K_{r}$, $P_r$, and $C_r$ the complete graph, the path, and the cycle on $r$ vertices, respectively, and
by $K_{r_{1},r_{2}}$ the complete bipartite graph where the one part has $r_{1}$ vertices and the other $r_{2}$.




\bigskip
\noindent
\textbf{Block-cut trees.}
A connected graph $G$ is \emph{biconnected} if for any $v \in V(G)$, $G \gm \{v\}$ is connected.
Notice that $K_{2}$ is the only  biconnected graph that it is not $2$-connected and that $K_1$ is not biconnected.
A \emph{block} of a graph $G$ is a maximal biconnected subgraph of $G$.
We name $\block(G)$ the set of all blocks of $G$ and we name $\cut(G)$ the set of all cut vertices of $G$.
If $G$ is connected, we define the \emph{block-cut tree} of $G$ to be the tree $\bct(G) = (V,E)$ such that
\begin{itemize}
\item  $V = \block(G) \cup \cut(G)$ and
\item $E = \{\{B,v\} \mid  B \in \block(G), v \in \cut(G) \cap V(B)\}$.
\end{itemize}
Note that  $\leaf(\bct(G)) \subseteq \block(G)$. It is worth mentioning that the block-cut tree of a graph can be computed in linear time using depth-first search~\cite{HopcroftT73}.
Basic properties of block-cut trees can be found, for instance, in~\cite{BondyM08} (where they are called \emph{block trees}).

\medskip
Let $\Fcal$ be a  set of connected graphs on at least two vertices. 
Given $H \in \Fcal$ and $B \in \leaf(\bct(H))$, we say that $(H,B)$ is an \emph{essential pair} if for each $H' \in \Fcal$ and each $B' \in \leaf(\bct(H'))$, $|E(B)| \leq |E(B')|$.
Given an essential pair $(H,B)$ of $\Fcal$, we define the \emph{first vertex} of $(H,B)$ to be, if it exists, the only cut vertex of $H$ contained in $V(B)$, or an arbitrarily chosen vertex of $V(B)$ otherwise.
We define the \emph{second vertex} of $(H,B)$ to be an arbitrarily chosen vertex of $V(B)$ that is a  neighbor in $H[B]$ of the first vertex of $(H,B)$.
Note that such a vertex always exists, as a block has at least two vertices.
Given an essential pair $(H,B)$ of $\Fcal$, we assume that the choices for the first and second vertices of $(H,B)$ are fixed.

Moreover, given an essential pair $(H,B)$ of $\Fcal$, we define the \emph{core} of $(H,B)$ to be the graph $H \gm (V(B) \sm \{a\})$ where $a$ is the first vertex of $(H,B)$.
Note that $a$ is a vertex of the core of $(H,B)$.


%

\bigskip
\noindent
\textbf{Minors and topological minors.}
Given two graphs $H$ and $G$ and two functions $\phi: V(H)\to V(G)$ and $\sigma: E(H)\to{\cal P}(G)$, we say that $(\phi,\sigma)$ is {\em a topological minor model of $H$ in $G$}
if
\begin{itemize}
\item for every $\{x,y\}\in E(H),$ $\sigma(\{x,y\})$ is an $(\phi(x),\phi(y))$-path 
  in $G$ and
\item if $P_{1},P_{2}$ are two distinct paths in $\sigma(E(H))$, then none of the internal vertices of $P_{1}$
  is a vertex of $P_{2}$.
\end{itemize}

The {\em branch} vertices of
$(\phi,\sigma)$ are the vertices in  $\phi(V(H))$, while the {\em subdivision} vertices of $(\phi,\sigma)$ are the internal vertices of the paths in $\sigma(E(H))$.

We say that $G$ contains $H$ as a \emph{topological minor}, denoted by $H\preceq_{\sf tm} G$, if there
is a topological minor model  $(\phi,\sigma)$ of $H$ in $G$.


Given two graphs $H$ and $G$ and a function $\phi: V(H)\to 2^{V(G)}$, we say that $\phi$ is {\em a minor model of $H$ in $G$}
if
\begin{itemize}
\item for every $x,y \in V(H)$ such that $x \not = y$, $\phi(x) \cap \phi(y) = \emptyset$,
\item for every $x \in V(H)$, $G[\phi(x)]$ is a connected non-empty graph and
\item for every $\{x,y\} \in E(H)$, there exist  $x' \in \phi(x)$ and $y' \in \phi(y)$ such that $\{x',y'\} \in E(G)$.
\end{itemize}

We say that $G$ contains $H$ as a \emph{minor}, denoted by $H\prem G$, if there
is a minor model  $\phi$ of $H$ in $G$.

\medskip
\noindent
\textbf{Graph separators and (topological) minors.}
Let $G$ be a graph and $S\subseteq V(G)$.
Then for each connected component $C$ of $G\gm S$, we define the \emph{cut-component} of the triple $(C,G,S)$ to be the graph
whose vertex set is $V(C) \cup S$ and whose edge set is
$E(G[V(C) \cup S])$.

\begin{lemma}
  \label{reifying}
  Let $i\geq 2$ be an integer, let $H$ be an $i$-connected graph,
  let $G$ be a graph, and
  let $S \subseteq V(G)$ such that $|S| \leq i-1$.
  If $H$ is a topological minor (resp. a minor) of $G$, then there exists a connected component $G'$ of $G\gm S$ such that
  $H$ is a topological minor (resp. a minor) of the cut-component of $(G',G,S)$.
\end{lemma}

\begin{proof}
  We prove the lemma for the topological minor version, and the minor version can be proved with similar arguments. Let $i$, $H$, $G$, and $S$ be defined as in the statement of the lemma.
  Assume that $H \pretp G$ and let
  $(\phi,\sigma)$ be a topological minor model of $H$ in $G$. If $S$ is not a separator of $G$, then the statement is trivial, as in that case the cut-component of $(G\setminus S,G,S)$ is  $G$. Suppose henceforth that $S$ is a separator of $G$, and assume for contradiction that there exist  two connected components $G_1$ and $G_2$ of $G\gm S$ and two distinct vertices $x_1$ and $x_2$ of $H$ such that
  $\phi(x_1) \in V(G_1)$ and $\phi(x_2) \in V(G_2)$.
  Then, as $H$ is $i$-connected, there should be $i$ internally vertex-disjoint paths from $\phi(x_1)$ to $\phi(x_2)$ in $G$.
  As $S$ is a separator of size at most $i-1$, this is not possible.
  Thus, there exists a connected component $G'$ of $G\gm S$ such that for each $x \in V(H)$, $\phi(x) \in V(G') \cup S$.
  This implies that $H$ is a topological minor of the cut-component of $(G',G,S)$.
\end{proof}

\begin{lemma}
  \label{captured}
  Let $G$ be a connected graph, let $v$ be a cut vertex of $G$, and let $V$ be the vertex set of a connected component of $G \gm \{v\}$.
  If $H$ is a connected graph such that $H \pretp G$ and for each leaf $B$ of $\mbox{\rm \bct}(H)$, $B \not \pretp G[V \cup \{v\}]$,
  then $H \pretp G \gm V$. 
\end{lemma}


\begin{proof}
  Let $G$, $v$, $V$, and $H$ be defined as in the statement of the lemma.
  Let $B \in \leaf(\bct(H))$.
  If $B$ is a single edge, then the condition $B \not \pretp G[V \cup \{v\}]$ implies that $V = \es$.
  But $V$ is the vertex set of a connected component of $G \gm \{v\}$ and so $V \not = \es$.
  This implies that the case $B$ is a single edge cannot occur.
  If $B$ is not a simple edge, then by definition $B$ is $2$-connected
  and then, by Lemma~\ref{reifying}, $B \pretp G \gm V$.
  This implies that there is a topological minor model $(\phi,\sigma)$ of $H$ in $G$  such that for each $B \in \leaf(\bct(H))$ and for each $b \in B$,
  $\phi(b)\not \in V$.

  We show now that for each $x \in V(H)$, $\phi(x)\not \in V$.
  If $ V(H) \sm (\bigcup_{B \in \leaf(\bct(H))} V(B)) = \es$ then the result is already proved.
  Otherwise, let $x \in V(H) \sm (\bigcup_{B \in \leaf(\bct(H))} V(B))$.
  By definition of the block-cut tree, there exist $b_1$ and $b_2$ in $\bigcup_{B \in \leaf(\bct(H))} V(B)$ such that $x$ lies on a $(b_1,b_2)$-path $P$ of $\Pcal(H)$.
  Let $P_i$ be the $(b_i,x)$-subpath of $P$ for each $i \in \{1,2\}$.
  By definition of $P$, we have that $V(P_1) \cap V(P_2) = \{x\}$.
  This implies that there exists a $(\phi(b_1),\phi(x))$-path $P'_1$ and a $(\phi(b_2),\phi(x))$-path $P'_2$ in $\Pcal(G)$ such that $V(P'_1) \cap V(P'_2) = \{\phi(x)\}$.
  Then, as $v$ is a cut vertex of $G$, it follows that
  $\phi(x) \not \in V$.
  Thus, for each $x \in V(H)$,  $\phi(x) \not \in V$.
  Let $\{x,y\}$ be an edge of $E(H)$.
  As $\sigma(\{x,y\})$ is a simple $(\phi(x),\phi(y))$-path, both $\phi(x)$ and $\phi(y)$ are not in $V$ and $v$ is a cut vertex of $G$, we have, with the same argumentation that before that,
  for each $z \in V(\sigma(\{x,y\})$, $z \not \in V$.
  This concludes the proof.
\end{proof}

Using the same kind of argumentation with minors instead of topological minors, we also obtain the following lemma.
\begin{lemma}
  \label{floppies}
  Let $G$ be a connected graph, let $v$ be a cut vertex of $G$, and let $V$ be the vertex set of a connected component of $G \gm \{v\}$.
  If $H$ is a graph such that $H \prem G$ and for each leaf $B$ of $\bct(H)$, $B \not \prem G[V \cup \{v\}]$,
  then $H \prem G \gm V$.
\end{lemma}

In the above two lemmas, we have required graph $H$ to be connected so that $\bct(H)$ is well-defined, but we could relax this requirement, and replace in both statements ``for each leaf $B$ of $\bct(H)$'' with ``for each connected component $H'$ of $H$ and each leaf $B$ of $\bct(H')$''.


%
%


\bigskip
\noindent
\textbf{Graph collections.}
Let ${\cal F}$  be a collection of graphs. From now on instead of ``collection of graphs''  we use the shortcut ``collection''.
If ${\cal F}$ is a non-empty finite collection and all its graphs are  non-empty, then we say that ${\cal F}$
is a {\em proper collection}. For any proper collection ${\cal F}$, we  define
${\sf size}({\cal F})=\max\{\{|V(H)|\mid H\in \cal F\}\cup\{|{\cal F}|\}\}$.
Note that if the size of ${\cal F}$ is bounded, then
the size of the graphs in ${\cal F}$ is also bounded.
We say that ${\cal F}$ is a {\em  planar collection} (resp. {\em planar subcubic collection}) if it is proper and
{\sl at least one} of the graphs in ${\cal F}$ is  planar (resp. planar and subcubic).
We say that ${\cal F}$ is a {\em connected collection} if it is proper and
{\sl all} the graphs in ${\cal F}$ are connected and of size at least two.
We  say that $\Fcal$ is a {\em (topological) minor antichain}
if no two of its elements are comparable via the  (topological)  minor relation.

Let $\Fcal$ be a proper collection. We extend the (topological) minor relation to $\Fcal$ such that, given a graph $G$,
$\Fcal \preceq_{\sf tm} G$ (resp. $\Fcal \preceq_{\sf m} G$) if and only if there exists a graph $H \in \Fcal$ such that $H \preceq_{\sf tm} G$ (resp. $H \preceq_{\sf m} G$).


\bigskip
\noindent\textbf{Tree and path decompositions.} A \emph{tree decomposition} of a graph $G$ is a pair ${\cal D}=(T,{\cal X})$, where $T$ is a tree
and ${\cal X}=\{X_{t}\mid t\in V(T)\}$ is a collection of subsets of $V(G)$
such that:
\begin{itemize}
\item $\bigcup_{t \in V(T)} X_t = V(G)$,
\item for every edge $\{u,v\} \in E$, there is a $t \in V(T)$ such that $\{u, v\} \subseteq X_t$, and
\item for each $\{x,y,z\} \subseteq V(T)$ such that $z$ lies on the unique path between $x$ and $y$ in $T$,  $X_x \cap X_y \subseteq X_z$.
\end{itemize}
We call the vertices of $T$ {\em nodes} of ${\cal D}$ and the sets in ${\cal X}$ {\em bags} of ${\cal D}$. The \emph{width} of a  tree decomposition ${\cal D}=(T,{\cal X})$ is $\max_{t \in V(T)} |X_t| - 1$.
The \emph{treewidth} of a graph $G$, denoted by $\tw(G)$, is the smallest integer $w$ such that there exists a tree decomposition of $G$ of width at most $w$.
For each $t \in V(T)$, we denote by $E_t$ the set $E(G[X_t])$.

A \emph{path decomposition} of a graph $G$ is a tree decomposition ${\cal D}=(T,{\cal X})$ of $G$ such that $T$ is a path, and the  \emph{pathwidth} of a graph $G$, denoted by $\pw(G)$, is the smallest integer $w$ such that there exists a path decomposition of $G$ of width at most $w$. Note that, by definition, for every graph $G$ it holds that $\pw(G) \geq \tw(G)$.

\bigskip
\noindent\textbf{Parameterized complexity.} We refer the reader to~\cite{DF13,CyganFKLMPPS15} for basic background on parameterized complexity, and we recall here only some very basic definitions.
A \emph{parameterized problem} is a language $L \subseteq \Sigma^* \times \mathbb{N}$.
Integer $k$ is called the \emph{parameter} of an instance $I=(x,k) \in \Sigma^* \times \mathbb{N}$.
A parameterized problem is \emph{fixed-parameter tractable} ({\sf FPT}) if there exists an algorithm $\Acal$, a computable function $f$, and a constant $c$ such that given an instance $I=(x,k)$,
$\Acal$ (called an {\sf FPT} \emph{algorithm}) correctly decides whether $I \in L$ in time bounded by $f(k) \cdot |I|^c$.


\bigskip
\noindent
\textbf{Definition of the problems.} Let ${\cal F}$ be a proper collection.
We define the parameter ${\bf tm}_{\cal F}$ as the function that maps graphs to non-negative integers as follows:
\begin{eqnarray}
  \label{overlaid}
  {\bf tm}_{\cal F}(G) & = & \min\{|S|\mid  S\subseteq V(G)\wedge   \Fcal \npreceq_{\sf tm}  G\setminus S   \}.
\end{eqnarray}
The parameter ${\bf m}_{\cal F}$ is defined analogously, just by replacing $\Fcal \npreceq_{\sf tm}  G\setminus S$ with $\Fcal \npreceq_{\sf m}  G\setminus S$.
The main objective of this paper is to study the problem of computing the parameters
${\bf tm}_{\cal F}$ and ${\bf m}_{\cal F}$  for graphs of bounded treewidth under several instantiations of
the collection ${\cal F}$. The corresponding decision problems are formally defined as follows.
\medskip\medskip\medskip

\hspace{-.6cm}\begin{minipage}{6.7cm}
  \paraprobl{\textsc{$\Fcal$-TM-Deletion}}
  {A graph $G$ and an integer $k\in \Nbb$.}
  {Is ${\bf tm}_\Fcal(G)\leq k$?}
  {The treewidth of $G$.}{6.7}
\end{minipage}~~~~~
\begin{minipage}{6.7cm}
  \paraprobl{\textsc{$\Fcal$-M-Deletion}}
  {A graph $G$ and an integer $k\in \Nbb$.}
  {Is ${\bf m}_\Fcal(G)\leq k$?}
  {The treewidth of $G$.}{6.7}
\end{minipage}
\medskip\medskip

Note that in both above problems,
we can always assume that ${\cal F}$ is an antichain
with respect to the considered relation.  Indeed,  if ${\cal F}$
contains two graphs $H_{1}$ and $H_{2}$ where $H_{1}\preceq_{\sf tm} H_{2}$, then
${\bf tm}_{\cal F}(G)={\bf tm}_{{\cal F}'}(G)$ where ${\cal F}'={\cal F}\setminus\{H_{2}\}$ (similarly
for the minor relation).

Throughout the article, we let $n$ and $\tw$ be the number of vertices and the treewidth of the input graph of the considered problem, respectively.

\section{Single-exponential lower bound for any connected $\Fcal$}
\label{struggle}



In this section we prove the following result.

\begin{theorem}
  \label{temporis}
  Let $\Fcal$ be a connected collection.
  Neither \pbtm nor \pbm can be solved in time $2^{\smallo(\tw)} \cdot n^{\Ocal(1)}$
  unless the \ETH fails.
\end{theorem}


\begin{proof}
  Let  $\Fcal$ be a connected collection. 
  We present a reduction from {\sc Vertex Cover} to \pbtm, both parameterized by the treewidth of the input graph, and then we explain the changes to be made to prove the lower bound for  \pbm.  {\sc Vertex Cover} cannot be solved in time $2^{\smallo(w)} \cdot n^{\Ocal(1)}$ unless the \ETH fails~\cite{ImpagliazzoP01} (in fact, it cannot be solved even in time $2^{\smallo(n)}$), where $w$ is the treewidth of the input graph. It is worth mentioning that our reduction bears some similarity with the classical reduction of Yannakakis~\cite{Yannakakis78} for general vertex-deletion problems.


  Without loss of generality, we can assume that $\Fcal$ is a topological minor antichain. First we select an essential pair $(H,B)$ of $\Fcal$.
  Let $a$ be the first vertex of $(H,B)$, $b$ be the second vertex of $(H,B)$, and $A$ be the core of $(H,B)$.
  For convenience, we also refer to $a$ and $b$ as vertices of the copies of $A$.

  Let $G$ be the input graph of the {\sc Vertex Cover} problem and let $<$ be an arbitrary total order on $V(G)$.
  We build a graph $G'$ starting from $G' = (V(G),\emptyset)$.
  For each vertex $v$ of $G$, we add a copy of $A$, which we call $A^v$, and we identify the vertices $v$ and $a$.
  For each edge $e = \{v,v'\} \in E(G)$ with $v < v'$, we remove $e$, we add a copy of $B$, which we call $B^e$, and we identify the vertices $v$ and $a$ and the vertices $v'$ and $b$.
  This concludes the construction of $G'$.
  Note that $|V(G')| = |V(G)| \cdot |V(A)| + |E(G)| \cdot |V(B) \sm \{a,b\}|$ and that $\tw(G') = \max\{\tw(G),\tw(H)\}$. For completeness, we provide a proof of the latter fact.

  For each $v \in V(G)$, we define
  $\mathcal{D}^v = (T^v,\mathcal{X}^v)$ to be a tree decomposition of $A^v$ and we fix $r_v \in V(T^v)$ such that $X^v_{r_v} \in \mathcal{X}^v$ contains the copy of $a$ in $A^v$.
  For each $e \in E(G)$, we define
  $\mathcal{D}^{e} = (T^{e},\mathcal{X}^{e})$ to be a tree decomposition of $B^{e}$ and we fix $r_{e} \in V(T^{e})$ such that $X^{e}_{r_{e}} \in \mathcal{X}^{e}$ contains the copy of $a$ and $b$ in $A^v$. We know that this bag exists as $\{a,b\} \in E(H)$.
  Let $\mathcal{D}^G = (T^G,\mathcal{X}^G)$ be a tree decomposition of $G$.
  We can then define a tree decomposition of $G'$ as follows.
  Start from
  $\mathcal{D} = (T,\mathcal{X})$ where
  $T$ is the union of $T^G$, of each $T^v$, $v \in V(G)$, and each $T^{e}$, $e \in E(G)$, and
  $\mathcal{X}$ is the union of $\mathcal{X}^G$, of each $\mathcal{X}^v$, $v \in V(G)$, and each $\mathcal{X}^{e}$, $e \in E(G)$.
  Then for each $v \in V(G)$, we arbitrarily choose $t_v \in V(T^G)$ such that $v \in \mathcal{X}^G_{t_v}$ and connect $t_v$ and $r_v$ in $T$.
  Then for each $e \in E(G)$, we arbitrarily choose $t_{e} \in V(T^G)$ such that $e \subseteq \mathcal{X}^G_{t_v}$ and connect $t_{e}$ and $r_{e}$ in $T$.
  This concludes the construction of a tree decomposition of $G'$.
  As for each $v \in V(G)$ and each $e \in E(G)$, the tree decompositions of $A^v$ and of $B^{e}$ are just smaller parts of a tree decomposition of $H$, we obtain that
  each bag of $\mathcal{D}$ is of size at most $\tw(G)$ if it comes from $\mathcal{X}^G$, or of size at most $\tw(H)$ if it comes from $\mathcal{X} \setminus \mathcal{X}^G$.
  Thus $\tw(G') = \max\{\tw(G),\tw(H)\}$.

  We claim that there exists a solution of size at most $k$ of {\sc Vertex Cover} in $G$ if and only if there is a solution of size at most $k$ of
  \pbtm
  in $G'$.

  In one direction, assume that $S$ is a solution of  
  \pbtm
  in $G'$ with $|S| \leq k$.
  By definition of the problem, for each $e=\{v,v'\} \in E(G)$ with $v < v'$,
  either  $B^e$  contains an element of $S$  or  $A^v$ contains an element of $S$.
  Let $S' =  \{v \in V(G) \mid \exists v' \in V(G) : v<v', e=\{v,v'\} \in E(G), (V(B^e) \sm \{v,v'\}) \cap S \not = \es\}
  \cup \{v \in V(G) \mid V(A^v) \cap S \not = \es\}
  $.
  Then $S'$ is a solution of  {\sc Vertex Cover} in $G$ and $|S'| \leq |S| \leq k$.

  In the other direction, assume that we have a solution $S$ of size at most $k$ of {\sc Vertex Cover} in $G$.
  We want to prove that $S$ is also a solution of 
  \pbtm
  in $G'$.
  For this, we fix an arbitrary $H' \in \Fcal$ and we show that $H'$ is not a topological minor of $G' \gm S$.
  First note that the connected components of $G' \gm S$ are either of the shape $A^v \gm \{v\}$ if $v \in S$, $B^e \gm e$ if $e \subseteq S$, or the union of $A^v$ with zero, one, or more graphs $B^{\{v,v'\}}\sm \{v'\}$ such that $\{v,v'\} \in E(G)$ if $v \in V(G) \sm S$.
  As $\Fcal$ is a topological minor antichain,
  for any $v \in V(G)$, $H'  \not \pretp A^v \gm \{v\}$ and for any $e \in E(G)$, $H' \not \pretp B^e \gm e$.
  Moreover, let $v \in V(G) \sm S$ and let $K$ be the connected component of $G' \gm S$ containing $v$.
  $K$ is the union of $A^v$ and of every $B^{\{v,v'\}}\gm \{v'\}$ such that $\{v,v'\} \in E(G)$.
  As, for each $v' \in V(G)$ such that $\{v,v'\} \in E(G)$, $v'$ is not an isolated vertex in $B^{\{v,v'\}}$, by definition of $B$, for any $B' \in \leaf(\bct(H'))$, $|E(B^{\{v,v'\}}\gm \{v'\})| < |E(B')|$.
  This implies that for each leaf $B'$ of $\bct(H')$ and for each $\{v,v'\} \in E(G)$,  $B' \not \pretp B^{\{v,v'\}}\gm \{v'\}$.
  It follows  by definition of $\Fcal$ that $H' \not \pretp A^v$.
  This implies by Lemma~\ref{captured} that $H'$ is not a topological minor of $K$. Moreover, as $H'$ is connected by hypothesis, it follows that that $H'$ is not a topological minor of $G' \gm S$ either. This concludes the proof for the topological minor version.

  Finally, note that the same proof applies to \pbm as well, just by replacing
  \pbtm with \pbm,
  topological minor with minor,
  $\pretp$ with $\prem$, and
  Lemma~\ref{captured} with  Lemma~\ref{floppies}.
\end{proof}

From Theorem~\ref{temporis} we can easily get the following corollary on planar graphs.

\begin{corollary}\label{patterns}
  Let $\Fcal$ be a connected planar collection.
  Neither \pbtm nor \pbm can be solved on planar graphs in time $2^{\smallo(\tw)} \cdot n^{\Ocal(1)}$
  unless the \ETH fails.
\end{corollary}
\begin{proof}
  We can assume that all the graphs in $\Fcal$ are planar, since an input planar graph $G$ does not contain any nonplanar graph as a (topological) minor. We reduce from
  {\sc Planar Vertex Cover} to \pbtm on planar graphs, both parameterized by the treewidth of the input graph, and the construction of $G'$ is the same as above. Note that since all the graphs in $\Fcal$ are planar, so is the essential pair $(H,B)$, and therefore the graph $G'$ is easily checked to be planar. Since {\sc Planar Vertex Cover} cannot be solved in time $2^{\smallo(w)} \cdot n^{\Ocal(1)}$ unless the \ETH fails~\cite{ImpagliazzoP01,Lic82}, where $w$ is the treewidth of the input graph, the result follows. Finally, the changes to be made for the minor version are the same as those in the proof of Theorem~\ref{temporis}.
\end{proof}

\section{Superexponential lower bounds}
\label{secsupexp}

Let $\Ccal$ be the set of all connected graphs that contain a block with at least five edges,
let $\mathcal{Q}$ be the set containing $P_5$ and all connected graphs that are {\sl not} minors of the \banner,
and let $\Scal = \{K_{1,s} \mid s \geq 4 \}$.
In this section,
we prove the following theorems. Note that, by definition, it holds that $\Ccal \subseteq \mathcal{Q}$, but we consider both sets because we will prove a stronger result for the set $\Ccal$ (Theorem~\ref{hardgene}, which applies to every subset of $\Ccal$) than for the set $\mathcal{Q}$ (Theorems~\ref{hardminor} and~\ref{hardtminor}, which apply to families containing a single graph $H$).




\begin{theorem}
  \label{hardgene}
  Let $\Fcal$ be a finite non-empty subset of $\Ccal$.
  Unless the \ETH fails, neither \textsc{$\Fcal$-M-Deletion} nor \textsc{$\Fcal$-TM-Deletion}  can be solved in time $2^{\smallo(\tw \log \tw)}\cdot n^{\Ocal(1)}$.
\end{theorem}

\begin{theorem}
  \label{hardminor}
  Let $H \in \mathcal{Q}$.
  Unless the \ETH fails,  \textsc{$\{H\}$-M-Deletion} cannot be solved in time $2^{\smallo(\tw \log \tw)}\cdot n^{\Ocal(1)}$.
\end{theorem}

\begin{theorem}
  \label{hardtminor}
  Let $H \in \mathcal{Q}\setminus \Scal$.
  Unless the \ETH fails,  \textsc{$\{H\}$-TM-Deletion} cannot be solved in time $2^{\smallo(\tw \log \tw)}\cdot n^{\Ocal(1)}$.
\end{theorem}

Note that if $H$ is a connected  graph such that $H \notin \mathcal{Q}$ (resp. $H \notin \mathcal{Q} \setminus \Scal$), then \textsc{$\{H\}$-M-Deletion} (resp. \textsc{$\{H\}$-TM-Deletion}) can be solved in time $2^{\Ocal(\tw)}\cdot n^{\Ocal(1)}$ by the single-exponential algorithms presented in~\cite{monster2}. On the other hand, if $H$ is a connected  (resp. planar subcubic) graph, then \textsc{$\{H\}$-M-Deletion} and \textsc{$\{H\}$-TM-Deletion}) can be solved in time $2^{\Ocal(\tw \log \tw)}\cdot n^{\Ocal(1)}$  by the algorithms presented in~\cite{monster1,SODA-arXiv,BasteST20-SODA}. In particular, note that these results altogether settle completely the asymptotic complexity of \textsc{$\{H\}$-M-Deletion} when $H$ is a connected graph; see Figure~\ref{shifting} for an illustration.



We first provide in Section~\ref{supexpgencons} a general framework that will be used in every reduction 
and then we explain how to modify this framework depending on the family $\Fcal$ we are considering.

\subsection{The general construction}
\label{supexpgencons}

In order to prove Theorem~\ref{hardgene}, Theorem~\ref{hardminor}, and Theorem~\ref{hardtminor}, we will provide reductions from the following problem, which is closely related to the \textsc{$k\times k$ Permutation Clique} problem defined by Lokshtanov et al.~\cite{permuclique}.

\paraprobl
{\sc $k\times k$ Permutation Independent Set}
{An integer $k$ and a graph $G$ with vertex set $\intv{1,k} \times \intv{1,k}$.}
{Is there an independent set of size $k$ in G with exactly one element from each row and exactly one element from each column?}
{$k$.}{14.3}

\begin{theorem}[Lokshtanov et al.~\cite{permuclique}]
  \label{borrowing}
  The {\sc $k\times k$ Permutation Independent Set} problem cannot be solved in time $2^{\smallo(k \log k)}$ unless the \ETH fails.
\end{theorem}

Let $\Fcal$ be a finite family of non-empty graphs.
The framework we are going to present  follows the ideas of the
construction given by Bonnet et al.~\cite{BonnetBKM-IPEC17}.
This framework mostly depends on $h := \min_{H \in \Fcal}|V(H)|$ but also on an integer $t_\Fcal$ whose value will be defined later.
Let $(G,k)$ be an instance of {\sc $k\times k$ Permutation Independent Set}.
As we are asking for an independent set that contains exactly one vertex in each row, we will assume without loss of generality that, for each pair $(i,j)$, $(i,j')$ in $V(G)$ with $j \neq j'$, $\{(i,j),(i,j')\} \in E(G)$.
We proceed to construct a graph $F$ that contains one gadget for each edge of the graph  $G$.
These gadgets are arranged in a cyclic way, separated by some other gadgets ensuring the consistency of the selected solution.

Formally, we first define the graph $K := K_{h-1}$.
For each $e \in E(G)$, and each $(i,j) \in \intv{1,k}^2$, we define the graph $B^e_{i,j}$ to be the disjoint union of $n_h$ copies of $K$, for some integer $n_h$, whose value will be $2$ in the minor case and ${h \choose 2}$ in the topological minor case, two new vertices  $a^e_{i,j}$ and  $b^e_{i,j}$, and
 $t_\Fcal$ other new vertices called \emph{$B$-extra vertices}. 
The graph $B^e_{i,j}$ is depicted in Figure~\ref{new-renegades}.

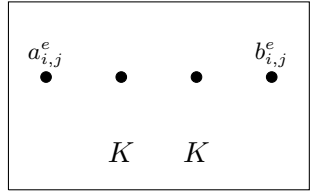
\begin{figure}[htb]
  \centering
  \scalebox{1}{\begin{tikzpicture}[scale=1]

      \node at (1,0) {$K$};
      \node at (2,0) {$K$};
      \node at (0,1.3) {\footnotesize{$a^{e}_{i,j}$}};
      \node at (3,1.3) {\footnotesize{$b^e_{i,j}$}};
      \vertex{0,1};
      \vertex{1,1};
      \vertex{2,1};
      \vertex{3,1};

      \draw (-0.5,-0.5) -- (3.5,-0.5) -- (3.5,2) -- (-0.5,2) -- (-0.5,-0.5);

    \end{tikzpicture}}
  \caption{The graph $B^e_{i,j}$ for $e\in E(G)$ and $(i,j) \in \intv{1,k}^2$ when $n_h=2$ and $t_\Fcal = 2$.}

  \label{new-renegades}
\end{figure}

Informally, the graph $B^e_{i,j}$, for every $e \in E(G)$, will play in $F$ the  role of the vertex $(i,j)$ in $G$.
For each $e \in E(G)$ and each $j \in \intv{1,k}$, we define the graph $C^e_j$ obtained from the disjoint union of every $B^e_{i,j}$, $i \in \intv{1,k}$, such that two graphs $B^e_{i_1,j}$ and $B^e_{i_2,j}$, $i_1 \not = i_2$, are {\em complete} to each other, that is, for every $i_1 \not = i_2$, if $v_1 \in V( B^e_{i_1,j})$ and $v_2 \in V(B^e_{i_2,j})$, then $\{v_1,v_2\} \in {E(C^e_{j})}$.
Informally, for a fixed $j \in \intv{1,k}$, the graph $C^e_j$, for every  $e \in E(G)$, corresponds to the column $j$ of $G$.
For every $e \in E(G)$, we also define the \emph{gadget graph} $D^e$ obtained from the disjoint union of every $C^e_{j}$, $j \in \intv{1,k}$, by adding, if $e = \{(i,j),(i',j')\}$, every edge $\{v_1,v_2\}$ such that $v_1 \in V(B^e_{i,j})$ and $v_2 \in V(B^e_{i',j'})$.
The graph $D^e$ is depicted in Figure~\ref{flattered}.

\begin{figure}[htb]
  \centering
  \scalebox{1}{\begin{tikzpicture}[scale=1]

      \draw[line width=3pt] (2,2) -- (4,4);
      \foreach \x in {1,2,3}
      {
        \node[circle] at (2*\x,7.5) {$C^e_{\x}$};

        \draw[line width=3pt] (2*\x,2) -- (2*\x,6);

        \draw[line width=3pt] (2*\x,2) -- (2*\x-0.6,2.6) -- (2*\x-0.6,5.4) -- (2*\x,6);
        \foreach \y in {1,2,3}
        {
          \node[circle,fill = white, minimum size=7.5mm] at (2*\x,2*\y) {$$};
          \node[circle] at (2*\x,2*\y) {$B^e_{\y,\x}$};
        }
      }
      \foreach \x in {1,2,3,4}
      {
        \draw[dotted] (2*\x-1,1)  -- (2*\x-1,8);
      }
      \draw[dotted] (1,7)  -- (7,7);
      \draw[dotted] (1,8)  -- (7,8);
      \draw[dotted] (1,1)  -- (7,1);

    \end{tikzpicture}}
  \caption{The gadget graph $D^e$ for $ e = \{(1,1),(2,2)\} \in E(G)$ where $k = 3$. A bold edge means that two graphs $B$ are connected in a complete bipartite way.}
  \label{flattered}
\end{figure}
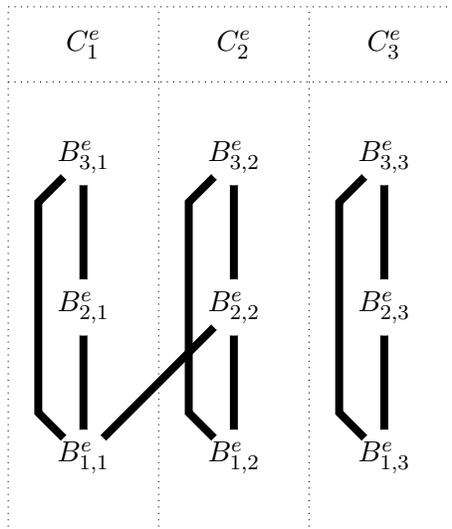

Informally, the graph $D^e$, $e\in E(G)$, encodes the edge $e$ of the graph $G$.
For every $e \in E(G)$, we also define $J^e$ such that $V(J^e) = \{c^e_{j} \mid j \in \intv{1,k}\} \cup \{r^e_{i} \mid i \in \intv{1,k}\}$ is a set of new vertices and $E(J^e) = \es$. It will be helpful to associate the $c^e_{j}$'s  with ``columns'' and the $r^{e}_{i}$'s with ``rows''.
Note that, in the following, $J^e$ may be enhanced, by adding vertices called \emph{$J$-extra vertices}, whose number depends on the family $\Fcal$ we are working with, but will always be linear in $k$.
The graphs $J^e$, $e \in E(G)$, are the \emph{separator gadgets} that will ensure the consistency of the selected solution.
Finally, the graph $F$ is obtained from the disjoint union of every $D^e$, $e \in E(G)$, and every $J^e$, $e \in E(G)$.
Moreover, we  fix a cyclic permutation $\sigma$ of the elements of $E(G)$, agreeing that  $\sigma^{-1}(e)$  and $\sigma(e)$ is the edge before and after $e$, respectively, in this cyclic ordering.
For each $e \in E(G)$, and each $(i,j) \in \intv{1,k}^2$, we add to $F$ the edges
$$
\{b^{\sigma^{-1}(e)}_{i,j},c^e_j\},
\{b^{\sigma^{-1}(e)}_{i,j},r^e_i\}, \{r^e_i,a^e_{i,j}\}, \mbox{~and~}
\{c^e_j,a_{i,j}^{e}\}.$$
This concludes the definition of the framework graph $F$, which is depicted in Figure~\ref{rejecting} (a similar figure appears in~\cite{BonnetBKM-IPEC17}).

\begin{figure}[htb]
  \centering
  \scalebox{1}{\begin{tikzpicture}[scale=1]

      \draw[rounded corners=4mm] (0,0) -- (10,0) -- (10,1) -- (0,1) -- (0,0);

      \node[circle, draw, fill=white] at (0,0) {$J^{e_1}$};
      \node[circle, draw, fill=white] at (2,0) {$D^{e_1}$};

      \node[circle, draw, fill=white] at (4,0) {$J^{e_2}$};
      \node[circle, draw, fill=white] at (6,0) {$D^{e_2}$};
      \node[circle, draw, fill=white] at (8,0) {$J^{e_3}$};
      \node[circle, draw, fill=white] at (10,0) {$D^{e_3}$};

    \end{tikzpicture}}
  \caption{
    The shape of the framework graph $F$ assuming that $k = 3$,  $G$ contains only the three edges $e_1$, $e_2$, and $e_3$, and $\sigma$ is the cyclic permutation $(e_1, e_2,e_3)$.}
  \label{rejecting}
\end{figure}
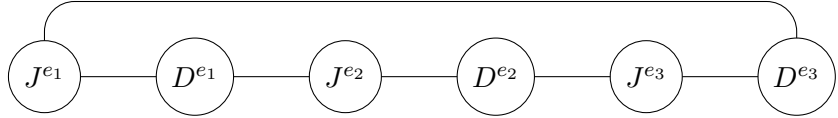
Note that in later constructions a gadget graph $D^e$, $e \in E(G)$, will only be connected to the (enhanced) separator gadgets $J^e$ and $J^{\sigma(e)}$ in a way that will be specified later and that depends on the family $\Fcal$.

Let $\ell := (n_h(h-1)+2+t_\Fcal)(k-1)k{m}$, where $m=|E(G)|$.
Note that $z = (n_h(h-1)+2+t_\Fcal)$ is the number of vertices of a $B^e_{i,j}$, $i,j \in \intv{1,k}$, $e \in E(G)$ and that
$\ell$ is the budget needed to select the vertex set of exactly $k-1$ graphs $B^e_{i,j}$, $i \in \intv{1,k}$ in each graph $C^e_j$, $j \in \intv{1,k}$, $e \in E(G)$.
The pair $(F,\ell)$ is called
the \emph{$\Fcal$-M-framework} of $(G,k)$ when $n_h = 2$, and
the \emph{$\Fcal$-TM-framework} of $(G,k)$ when $n_h = {h \choose 2}$.
When the value $n_h$ is not relevant, the pair $(F,\ell)$ is simply called the \emph{$\Fcal$-framework} of $(G,k)$.
For convenience, we always assume some prespecified permutation $\sigma$ associated with the graph $F$.

For each family $\Fcal$, given an input $(G,k)$ of {\sc $k\times k$ Permutation Independent Set}, we will consider $(F,\ell)$, the $\Fcal$-framework of $(G,k)$, and create another pair $(F_\Fcal,\ell)$, called the \emph{enhanced $\Fcal$-framework}, where $F_\Fcal$ is a graph obtained from $F$ by adding some new vertices and edges.
The added vertices  will be $B$-extra vertices 
{or} {$J$-extra vertices}. 
The added edges will  be either inside some (enhanced) $B_{i,j}^{e}$, or from some $D^{e}$ to the (enhanced) $J^{e}$ and $J^{\sigma(e)}$. More formally, the additional edges will be from the set
 $$\Big(\bigcup_{e \in E(G),\atop  i,j\in \intv{1,k}} V(B^e_{i,j}) \times V(B^e_{i,j})\Big) \cup \Big(\bigcup_{e \in E(G)} V(D^e) \times V(J^e) \Big)\cup \Big(\bigcup_{e \in E(G)} V(D^e) \times V(J^{\sigma(e)})\Big),$$
 by interpreting ordered pairs as edges.
Note that $F$ will always be a subgraph of $F_\Fcal$.
We will claim that there exists a solution of {\sc $k\times k$ Permutation Independent Set} on $(G,k)$ if and only if there exists a solution of \textsc{$\Fcal$-M-Deletion} (resp. \textsc{$\Fcal$-TM-Deletion}) on {$(F_\Fcal,\ell)$}.
In order to do this, we first prove a generic lemma, namely Lemma~\ref{reimitation}, and then provide a property, namely Property~\ref{new-wealthierm} (resp. Property~\ref{new-wealthiertm}), which we will prove for each family $\Fcal$ depending on the enhanced $\Fcal$-framework $F_\Fcal$

Let us now provide an upper bound on the treewidth (in fact, the pathwidth) of $F_\Fcal$.
Let $\{e_1,\ldots, e_m\} = E(G)$ such that for each $i \in \intv{1,m}$, $\sigma(e_i) = e_{i+1}$ with the convention that $e_{m+1} = e_1$.
First note that, for each $e \in E(G)$, the set $V(J^e) \cup V(J^{\sigma(e)})$ disconnects the vertex set $V(D^e)$ from the rest of $F_\Fcal$.
Moreover, if $e = \{(i,j),(i',j')\}$, then the bags
\begin{itemize}
\item[]
  $V(J^{e_1}) \cup V(J^e) \cup V(J^{\sigma(e)}) \cup V(C^e_1) \cup
  V(B^e_{i,j}) \cup V(B^e_{i',j'})$,
\item[]
  $V(J^{e_1}) \cup V(J^e) \cup V(J^{\sigma(e)}) \cup V(C^e_2) \cup
  V(B^e_{i,j}) \cup V(B^e_{i',j'})$, $\ldots$,
\item[]
  $V(J^{e_1}) \cup V(J^e) \cup V(J^{\sigma(e)}) \cup V(C^e_k) \cup
  V(B^e_{i,j}) \cup V(B^e_{i',j'})$
\end{itemize}
form a path decomposition of $G[V(J^e) \cup V(J^{\sigma(e)}) \cup V(D^e)]$ of width $(n_h(h-1)+2+t_\Fcal)(k+1)-1 + 3\cdot |V(J^{e_1})|$ (using the fact that $|V(J^{e})| = |V(J^{e_1})|$ for each $e \in E(G)$).
Let denote by $P^e$ this decomposition.
By concatenating the path decompositions $P^{e_1}$, $P^{e_2}$, $\ldots$, and $P^{e_m}$, we obtain a path decomposition of $F_\Fcal$ whose width, by using the fact that $|V(J^{e_1})| = \Ocal(k)$, is linear in $k$.

We start by proving that for each column $C_j^e$, $e \in E(G)$, $j \in \intv{1,k}$, containing a minimum number of vertices of the solution, all the remaining vertices belong to the same row of this column.

%

\begin{lemma}
  \label{colu}
  Let $\Fcal$ be a family of graphs
  and let $(F,\ell)$ be the $\Fcal$-TM-framework of an input $(G,k)$ of {\sc $k\times k$ Permutation Independent Set}.
  Let  $S$ be a solution of  \textsc{$\Fcal$-TM-Deletion} on $(F,\ell)$
  and let $e \in E(G)$ and $j \in \intv{1,k}$  such that
  the quantity $|V(C^e_j) \sm S|$ is maximized, i.e.,
  for all $e' \in E(G)$ and $j' \in \intv{1,k}$\  $|V(C^e_j) \sm S|\geq |V(C^{e'}_{j'}) \sm S|$.
 Then there exists $i \in \intv{1,k}$ such that $V(C^e_j) \sm S \subseteq V(B^e_{i,j})$.
\end{lemma}

\begin{proof}
We set $z=({h \choose 2}(h-1)+2+t_\Fcal)$ and observe that $|S|\leq \ell=z(k-1)km$.
  Let  $h = \min_{H \in \Fcal}|V(H)|$, and note that we can always assume that $h \geq 2$ as otherwise the problem is trivial. Choose $e \in E(G)$ and $j \in \intv{1,k}$  so that
  the quantity $|V(C^e_j) \sm S|$ is maximized.
  In order to prove the lemma, we show that
the assumption that
   there exist $i_1,i_2 \in \intv{1,k}$, with $i_1 \not = i_2$ such that $(V(C^e_j) \sm S) \cap V(B^e_{i_1,j}) \not = \es$ and $(V(C^e_j) \sm S) \cap V(B^e_{i_2,j}) \not = \es$ implies that  $F \gm S$ contains $K_h$  as a topological minor, for any value of $t_\Fcal \in \Nbb$.

%

  We claim that $|V(C^e_j) \sm S|\geq z$. Indeed, if $|V(C^e_j) \sm S|<z$,  by the maximality in the choice of $e$ and $j$, it follows that
  $|V(C^{e'}_{j'}) \sm S|\leq z-1$, for all $e'\in E(G)$ and $j'\in[1,k]$.
 This implies that $|S\cap V(C_{j'}^{e'})|\geq  |V(C_{j'}^{e'})|-(z-1)=kz-z+1=z(k-1)+1$  for all $e'\in E(G)$ and $j'\in[1,k]$.
 As there are $m$ choices for $e'$ and $k$ choices for $i'$ we have that
$$|S|\geq \bigcup_{e'\in E(G)\atop i'\in[1,k]}|S\cap V(C_{j'}^{e'})|\geq (z(k-1)+1)km>z(k-1)km=\ell,$$
  a contradiction. We just proved that
  \begin{eqnarray}
  |V(C^e_j) \sm S| & \geq  & {h \choose 2}(h-1)+2+t_\Fcal.\label{sdafdfds}
  \end{eqnarray}

We next pick $U_M$ as a set $V(B^e_{i,j}) \sm S$, $i \in \intv{1,k}$, with the maximum number of elements.  We  claim that, if $|U_M| \leq h-1$, then $F\gm S$ contains $K_{h}$ as a topological minor. We fist observe  that at least $h$ of the sets in ${\cal Z}=\{V(B^e_{i,j}) \sm S\mid i \in \intv{1,k\}}$ are non-empty. To verify this, suppose to the contrary that the set  ${\cal Z}'$, consisting of  the non-empty elements of ${\cal Z}$, has cardinality at  most $h-1$.
By the maximality of the choice of $U_M$, we obtain that each set in ${\cal Z}'$
has at most $|U_M|$ elements. We then observe that
 $|V(C^e_j) \sm S|=|\cupall {\cal Z}'|=(h-1)\cdot |U_M|\leq (h-1)^2$, a contradiction to~\eqref{sdafdfds},
 as $ (h-1)^2<{h \choose 2}(h-1)+2+t_\Fcal$.
We just proved that $|{\cal Z}'|\geq h$. By picking one vertex from each set in ${\cal Z}'$, we conclude that $F \gm S$ contains a clique of size $h$ as a subgraph and the claim follows.\medskip

From now on, we assume that  $h \leq |U_M|$. In fact, we claim that
if $|U_M| \leq {h \choose 2}(h-1)+2+t_\Fcal - {h \choose 2}$, then $F\gm S$ contains $K_{h}$ as a topological minor.
For this,
 let $Q$ be any set of $h$ vertices of $U_M$ and let $Z$ be a set of ${h \choose 2}$ vertices of $V(C^e_j) \sm (S \cup U_M)$ (this set exists because of~\eqref{sdafdfds}).
  As each vertex of $Z$ is a neighbor of each vertex of $Q$, we obtain that $F[Q \cup Z]$, which is a subgraph of $F\setminus S$, contains $K_h$ as a topological minor.\medskip

According to the previous claim, we can assume that $|U_M| > {h \choose 2}(h-1)+2+t_\Fcal - {h \choose 2}$ or, equivalently
  \begin{eqnarray}
  {h \choose 2}(h-2)+2+t_\Fcal < |U_{M}|.\label{sddsfads}
  \end{eqnarray}

Let $a$ be the size of the largest clique $K_a$ in $F[U_M]$.
We claim that $a\geq h-1$. For this, observe that if the biggest clique in $F[U_M]$ has size at most $h-2$, then $|U_{M}|\leq  {h \choose 2}(h-2)+2+t_\Fcal$, contradicting~\eqref{sddsfads}. \medskip

We just derived that $F[U_M]$ contains a clique $K_{h-1}$.  By our initial assumption, we have
that  $U_A=(V(C^e_j) \sm S) \cap V(B^e_{i',j}) \not = \es$ for some $i'\neq i$.
By combining $K_a$ with any vertex of $U_A$, we obtain $K_h$ as a subgraph of $F \gm S$,
%
and the lemma follows.
\end{proof}

Note that Lemma~\ref{colu} is also valid for the minor version with the same $\Fcal$-TM-framework.
However, for our future constructions, we need this statement to hold for the (smaller)  $\Fcal$-M-framework as well, where $n_{h}=2$.
%


%
\begin{lemma}
  \label{colu2}
  Let $\Fcal$ be a family of graphs
  and let $(F,\ell)$ be the $\Fcal$-M-framework of an input $(G,k)$ of {\sc $k\times k$ Permutation Independent Set}.
  Let  $S$ be a solution of  \textsc{$\Fcal$-M-Deletion} on $(F,\ell)$
  and let $e \in E(G)$ and $j \in \intv{1,k}$  such that
  the quantity $|V(C^e_j) \sm S|$ is maximized, i.e.,
  for all $e' \in E(G)$ and $j' \in \intv{1,k}$\  $|V(C^e_j) \sm S|\geq |V(C^{e'}_{j'}) \sm S|$.
 Then there exists $i \in \intv{1,k}$ such that $V(C^e_j) \sm S \subseteq V(B^e_{i,j})$.
\end{lemma}

\begin{proof}
  Let  $h = \min_{H \in \Fcal}|V(H)|$.
  Recall that $\ell = (2(h-1)+2+t_\Fcal)(k-1)km$.
  Again, in order to prove the lemma, we show that
the assumption that
   there exist $i_1,i_2 \in \intv{1,k}$, with $i_1 \not = i_2$ such that $(V(C^e_j) \sm S) \cap V(B^e_{i_1,j}) \not = \es$ and $(V(C^e_j) \sm S) \cap V(B^e_{i_2,j}) \not = \es$ implies that  $F \gm S$ contains $K_h$  as a  minor, for any value of $t_\Fcal \in \Nbb$.

  Let us fix the value of $t_\Fcal \in \Nbb$.
  Let $e \in E(G)$ and $j \in \intv{1,k}$ be such that  $|V(C^e_j) \sm S|$ is maximized.\medskip


  We claim that $|V(C^e_j) \sm S|\geq 2h+t_\Fcal$. Indeed, if $|V(C^e_j) \sm S|<2h+t_\Fcal$,  by the maximality in the choice of $e$ and $j$, it follows that
  $|V(C^{e'}_{j'}) \sm S|\leq 2h+t_\Fcal-1$, for all $e'\in E(G)$ and $j'\in[1,k]$.
 This implies that $|S\cap V(C_{j'}^{e'})|\geq  |V(C_{j'}^{e'})|-(2h+t_\Fcal-1)=k(2(h-1)+t_\Fcal+2)-(2h+t_\Fcal-1)=(2h+t_{\Fcal})(k-1)+1$  for all $e'\in E(G)$ and $j'\in[1,k]$.
 As there are $m$ choices for $e'$ and $k$ choices for $i'$ we have that
$$|S|\geq \bigcup_{e'\in E(G)\atop i'\in[1,k]}|S\cap V(C_{j'}^{e'})|\geq ((2h+t_{\Fcal})(k-1)+1)km>(2h+t_{\Fcal})(k-1)km=\ell,$$
  a contradiction. We just proved that
  \begin{eqnarray}
  |V(C^e_j) \sm S| & \geq  &2h+t_\Fcal.\label{sdafdfdss}
  \end{eqnarray}

  An edge $e\in E(C^e_j)$ is \emph{transversal} if
  there is no $i\in[1,k]$ such that both endpoints of $e$ belong to $B^e_{i,j}$.
%
  The important property of a transversal edge $e = \{v_1,v_2\}$ is that $N_{C^e_j}(\{v_1,v_2\}) = V(C^e_j) \sm \{v_1,v_2\}$.
  A \emph{transversal matching} of $C^e_j$ is a matching that contains only transversal edges.
  Note that if there exists a transversal matching $M$ of size $h$ over a set of  vertices $T \subseteq V(C^e_j)\setminus S$, then, by contracting every edge of $M$, it follows that $C^e_j[T]$, and therefore  $F\setminus  S$ as well, contains $K_h$ as a minor.

  Let $U_M$ be a set $V(B^e_{i,j}) \sm S$, $i \in \intv{1,k}$, with the maximum number of elements.
  If $|U_M| \leq h+t_\Fcal$, then, because of~\eqref{sdafdfdss}, the graph $C_j^e \gm S$ contains a transversal matching of size at least $h$.
  This can be seen, for instance, by considering the complete $k$-partite graph where each part contains the vertices in $V(B^e_{i,j}) \sm S$, for $i \in \intv{1,k}$, and noting that it admits a perfect matching (which defines a transversal matching in $C_j^e \gm S$ of size at least $h$) by applying Tutte's criterion~\cite{Die10} on the existence of a perfect matching in a general graph (recall that a graph $G$ contains a perfect matching if and only if there is no set $S \subseteq V(G)$ whose removal generates more than $|S|$ odd-sized components). Thus $C_j^e \gm S$, and therefore  $F\setminus  S$ as well, contains a clique of  $h$ vertices as a minor.

Assume now that $h+t_\Fcal < |U_M|$.
Let $a$ be the maximum size of a clique in
$C^e_j[U_M]$.  {As $h\geq 1$} we have that
$|U_M| \geq t_\Fcal+2$, therefore
\begin{eqnarray}
a & \geq & \left\lceil \frac{|U_M|-(t_\Fcal+2)}{2} \right\rceil.\label{ki6ykl}
\end{eqnarray}

  We claim that, if $|U_M|< 2h+t_\Fcal$, then $F\setminus S$ contains a clique of $h$ vertices as a minor. For this, we set $U_A := V(C^e_j) \sm (S \cup U_M)$ and
we distinguish two cases, depending on the parity of the quantity  $|U_M|-t_{\cal F}$.\medskip

\noindent{\sl Case 1:}  $|U_M|  = t_\Fcal + 2u$, with $h < 2u < 2h$. Then, from~\eqref{ki6ykl},
$a\geq (2u-2)/2$, therefore $C^e_j[U_M]$ contains a clique $K^*$ of size $u-1$ while the vertices
of  $C^e_j[U_M]$ that are not in $K^*$ are
$t_\Fcal+u+1$. Moreover, by~\eqref{sdafdfdss},
  $ |U_A| \geq  (2h+t_\Fcal) -(t_\Fcal + 2u) = 2h-2u$, so $C^e_j\gm (S \cup V(K^*))$ contains a transversal matching of size at least $q:=\min\{2h-2u, t_\Fcal+u+1\}$, which, when contracted,
  creates a  clique $K^+$ of size at least $q$ whose  vertices are connected with all $u-1$ vertices of $K^*$.
Also, using the inequality $h < 2u < 2h$, we  obtain
\begin{eqnarray*}
 (u-1) + (t_\Fcal+u+1) \geq 2u & >&  h\mbox{~and~}\\
 (u-1) + (2h-2u)  =   2h-(u+1)   & \geq & h,
 \end{eqnarray*}
 therefore, in any case, $(u-1)+q\geq h$.
  Thus, by taking $K^*$ with $K^+$ and all the edges between them, we deduce that $C^e_j\gm S$, and therefore $F\setminus S$ as well, contains a clique of size at least $h$ as a minor.

\medskip

\noindent{\sl Case 2:} $|U_M|  = t_\Fcal + 2u+1$, with $h < 2u+1 < 2h$. Then $C^e_j[U_M]$ contains a clique $K^*$ of size $u$ and $t_\Fcal+u+1$ vertices outside this clique $K^*$. Moreover, again by~\eqref{sdafdfdss},
  $|U_A| \geq 2h-2u-1$, so $C^e_j\gm (S \cup V(K^*))$ contains a transversal matching of size at least $q:=\min\{2h-2u-1, t_\Fcal+u+1\}$.
  On the other hand, using the inequality $h < 2u+1 < 2h$, we know that
  \vspace{-9mm}

 \begin{eqnarray*}
  u + (t_\Fcal+u+1) \geq 2u+1 & > & h \mbox{~and~} \\
  u + (2h-2u-1) = 2h-(u+1) & \geq & h.
  \end{eqnarray*}
  Thus $u+q\geq h$, and, as in the previous case,   we deduce that $C^e_j\gm S$, and therefore $F\setminus S$ as well,  contains a clique of size $h$ as a minor. The claim follows.\medskip

  Therefore, what remains is to examine the case where $|U_M| \geq 2h+t_\Fcal$. In this case, because of \eqref{ki6ykl},  $C^e_j[U_M]$  contains a clique of size $a\geq (2h-2)/{2}=h-1$. Combining this clique with any vertex in a set
$(V(C^e_j) \sm S) \cap V(B^e_{i',j})$ that, because of our initial assumption, is non-empty for some $i'\neq i$,
 we obtain $K_h$ as a subgraph of $F \gm S$, and the lemma follows.
\end{proof}

The purpose of Lemma~\ref{colu} and Lemma~\ref{colu2} is to obtain Lemma~\ref{reimitation} that states that for any solution $S$ of
 \textsc{$\Fcal$-M-Deletion} (resp. \textsc{$\Fcal$-TM-Deletion}) and for any
 $B^e_{i,j}$, $e \in E(G)$ and $(i,j) \in \intv{1,k}^2$, either $V(B^e_{i,j}) \cap S = \es$
 or {$V(B^e_{i,j}) \subseteq S$}. Moreover, there is exactly one $B^e_{i,j}$ such that
  $V(B^e_{i,j}) \cap S = \es$ in each column $C^e_j$,  $e \in E(G)$, $j \in \intv{1,k}$.

\begin{lemma}
  \label{reimitation}
  Let $\Fcal$ be a family of graphs
  and let $(F,\ell)$ be the $\Fcal$-M-framework (resp. $\Fcal$-TM-framework)  of an input $(G,k)$ of {\sc $k\times k$ Permutation Independent Set}.
  For every solution  $S$ of  \textsc{$\Fcal$-M-Deletion} (resp. \textsc{$\Fcal$-TM-Deletion}) on $(F,\ell)$,
  for every $e \in E(G)$ and  every $j \in \intv{1,k}$, there exists $i \in \intv{1,k}$ such that $V(C^e_j)\sm S = V(B^e_{i,j})$.
  Moreover, for every $e \in E(G)$, $V(J^e) \cap S = \es$.
\end{lemma}

\begin{proof}
  Let $S$ be a solution of \textsc{$\Fcal$-M-Deletion} (resp. \textsc{$\Fcal$-TM-Deletion}) on
  $(F,(n_h(h-1) +2 + t_\Fcal)(k-1)km)$.
  By Lemma~\ref{colu2} (resp. Lemma~\ref{colu}), we know that for every $e \in E(G)$ and every $j \in \intv{1,k}$, there is some $i$ such that for every $i'\in[1,k]\setminus\{i\}$, $B_{i',j}^{e}\subseteq S$.
  As each $B_{i,j}^{e}$ has  $n_h(h-1) +2 + t_\Fcal$ vertices, we obtain that
  $|V(C^e_j) \cap S| \geq (n_h(h-1) +2 + t_\Fcal)(k-1)$.
  As there are exactly $m$ edges and $k$ columns, the budget is tight and we obtain that
  $|V(C^e_j) \cap S| = (n_h(h-1) +2 + t_\Fcal)(k-1)$.
  This implies that
  $|V(C^e_j) \sm S| = n_h(h-1) +2 + t_\Fcal$,
  corresponding to the size of a set $B^e_{i,j}$, for some $i \in \intv{1,k}$.
  Moreover, as all the vertices of $S$ are vertices in the sets $D^e$, we also have that for every $e \in E(G)$, $V(J^e) \cap S = \es$.
  The lemma follows.
\end{proof}

Using Lemma~\ref{reimitation}, for every edge $e\in E(G)$, it will be possible to make a correspondence between a permutation corresponding to a solution of {\sc $k\times k$ Permutation Independent Set} and the $k$ pairs $(i,j)$ in $\intv{1,k}^2$ for which  $V(B^e_{i,j}) \cap S = \es$. In order to ensure the consistency of the selected solution among the gadget graphs $D^e$, $e \in E(G)$,  we need to show that given  $(i,j) \in \intv{1,k}^2$, if there exists $e \in E(G)$ such that $V(B^e_{i,j}) \cap S = \es$, then for every $e' \in E(G)$, we have $V(B^{e'}_{i,j}) \cap S = \es$.
For this, we state two properties, namely Property~\ref{new-wealthierm} and Property~\ref{new-wealthiertm}, applying to the minor and topological minor version of the problem, respectively.
Then we prove Lemma~\ref{lemmasol}, stating that if the corresponding property  holds, then we indeed have the desired consistency for the corresponding problem, which allows to find a  solution of \textsc{$k\times k$ Permutation Independent Set}.

\begin{prope}
  \label{new-wealthierm}
  Let $\Fcal$ be a family of graphs
  and let $(F_\Fcal,\ell)$ be the enhanced $\Fcal$-M-framework of an input $(G,k)$ of {\sc $k\times k$ Permutation Independent Set}.
  Let $S$ be a solution
  of  \textsc{$\Fcal$-M-Deletion} on $(F_\Fcal,\ell)$.
  For every $e \in E(G)$, and for every $i,j \in \intv{1,k}$, if $b^e_{i,j} \not \in S$ then for every $i'\in \intv{1,k} \sm \{i\}$, we have  $a^{\sigma(e)}_{i',j} \in S$.
\end{prope}

\begin{prope}
  \label{new-wealthiertm}
  Let $\Fcal$ be a family of graphs
  and let $(F_\Fcal,\ell)$ be the enhanced $\Fcal$-TM-framework of an input $(G,k)$ of {\sc $k\times k$ Permutation Independent Set}.
  Let $S$ be a solution
  of  \textsc{$\Fcal$-TM-Deletion} on $(F_\Fcal,\ell)$.
  For every $e \in E(G)$, and for every $i,j \in \intv{1,k}$, if $b^e_{i,j} \not \in S$ then for every $i'\in \intv{1,k} \sm \{i\}$, we have  $a^{\sigma(e)}_{i',j} \in S$.
\end{prope}

The above properties state that the choices of the vertices $a_{i,j}^e, b_{i,j}^e$ are consistent through the  graph $F_\Fcal$.

\begin{lemma}
  \label{lemmasol}
  Let $\Fcal$ be a family of graphs, let $(F_\Fcal,\ell)$ be the enhanced $\Fcal$-M-framework (resp. enhanced $\Fcal$-TM-framework) of an input $(G,k)$ of {\sc $k\times k$ Permutation Independent Set}.
  If Property~\ref{new-wealthierm} (resp. Property~\ref{new-wealthiertm}) holds and there exists a solution $S$ of
  \textsc{$\Fcal$-M-Deletion} (resp. \textsc{$\Fcal$-TM-Deletion}) on $(F_\Fcal,\ell)$, 
  then, for any $e \in E(G)$, the set $T^e=\{(i,j) \mid V(B^e_{i,j}) \cap S = \es\}$ is a solution of \textsc{$k\times k$ Permutation Independent Set} on $(G,k)$. Moreover, for   any $e_1,e_2 \in E(G)$, $T^{e_1} = T^{e_2}$.
\end{lemma}

\begin{proof}
  Let $S$ be a solution of \textsc{$\Fcal$-M-Deletion} (resp. \textsc{$\Fcal$-TM-Deletion})
  on $(F_\Fcal,\ell)$.
  Note that this implies that $S$ is also a solution of \textsc{$\Fcal$-M-Deletion} (resp. \textsc{$\Fcal$-TM-Deletion}) on $(F,\ell)$, the $\Fcal$-M-framework (resp.  $\Fcal$-TM-framework) of $(G,k)$, and so, Lemma~\ref{reimitation} can be applied.
  Let $\sigma$ be the cyclic permutation associated with  $F_\Fcal$.
  For each $e \in E(G)$, let $T^e =  \{(i,j) \mid V(B^e_{i,j}) \cap S = \es\}$.
  By Lemma~\ref{reimitation},  for each $e \in E(G)$, $T^e$ contains exactly one element from each column.
  We first show that for any $e \in E(G)$, $T^e = T^{\sigma(e)}$.
  Let $e \in E(G)$ and
  let $(i,j) \in T^e$.
  As  $(i,j) \in T^e$, we have that $b^e_{i,j} \not \in S$.
  By Property~\ref{new-wealthierm} (resp. Property~\ref{new-wealthiertm}), for each $i' \in \intv{1,k} \sm \{i\}$ it holds that $a^{\sigma(e)}_{i',j} \in S$, and thus $(i',j) \not \in T^{\sigma(e)}$.
  As  $T^{\sigma(e)}$ contains exactly one element from each column
  it follows that $(i,j)  \in T^{\sigma(e)}$. 
  As both $T^e$ and $T^{\sigma(e)}$ are of size exactly $k$, we obtain that $T^e = T^{\sigma(e)}$.

  By repeating the above argument iterating cyclically along the permutation $\sigma$, we obtain that for any $e_1,e_2 \in E(G)$, $T^{e_1} = T^{e_2}$.
  Let $(i,j),(i',j') \in T^{e_{1}}$.
  The existence of an edge $e = \{(i,j),(i',j')\}$ in $E(G)$ implies that in $D^e$, there is a vertex in $B^e_{i',j'}$ (in fact, any vertex of $B^e_{i',j'}$) that is fully connected to a copy of $K = K_{h-1}$  that is in $B^e_{i,j}$ (in fact, to all such copies), and so $D^e$ contains the clique $K_h$ as a subgraph.
  As $h = \min_{H \in \Fcal}|V(H)|$, this is not possible, and therefore {$T^{e}$} is an independent set in $G$ of size $k$.
  Moreover, by the construction of $G$, {$T^e$} contains at most one vertex per row and by Lemma~\ref{reimitation}, it contains exactly one vertex per column.
  The lemma follows.
\end{proof}

Given a solution $P$ of  {\sc $k\times k$ Permutation Independent Set} on $(G,k)$, we define $S_P = \{v \in V(F_\Fcal) \mid v \in B^e_{i,j} : e \in E(G), (i,j) \in \intv{1,k}^2\sm P\}$, where $(F_\Fcal,\ell)$ is the enhanced $\Fcal$-framework of $(G,k)$.
Note that $|S_P| = \ell$. In what follows we will prove that $S_{P}$
is a solution of
  \textsc{$\Fcal$-M-Deletion} (or \textsc{$\Fcal$-TM-Deletion}) on $(F_\Fcal,\ell)$ for each instantiation of   $F_{\cal F}$ that we will consider.

We now proceed to describe how to complete, starting from $F$, the construction of the  enhanced $\Fcal$-M-framework (or enhanced $\Fcal$-TM-framework)
$F_\Fcal$, depending on ${\cal F}$, towards proving Theorems~\ref{hardgene},~\ref{hardminor}, and~\ref{hardtminor}.

\subsection{The reduction for $\Fcal \subseteq \Ccal$}
\label{globred}

In order to prove Theorem~\ref{hardgene}, we will need some extra definitions.

Given a finite graph $H$, we define the \emph{block edge size function} $\besf_H: \Nbb \to \Nbb$ to be such that for any $x \in \Nbb$, $\besf_H(x)$ equals the number of edges of $H$ that are contained in a block with  at least $x$ edges.
Note that this function is a decreasing function and, as we only deal with finite graphs, for any finite graph $H$, there exists $x_0 \in \Nbb$ such that $\besf_H(x_0) = 0$ (notice that the minimum such $x_0$ is one more than the maximum number of edges of a block of $H$).
Given two block edge size functions $f$ and $g$, corresponding to two graphs, we say that $f \prec g$ if there exists an $x_0$, called a \emph{witness} of the inequality, such that $f(x_0) < g(x_0)$ and for each $x \geq x_0$, $f(x) \leq g(x)$.
It can be verified that $\prec$ is a {total} order on the set $\{\besf_H\mid H\mbox{~is a finite graph}\}$.
Note also that given two graphs $H$ and $H'$, if $\besf_H \prec \besf_{H'}$ then $H'$ cannot be a minor of $H$.
Intuitively, if one considers only the blocks of $H'$ with at least $x_0$ edges, then there are too many edges to  fit within the blocks of $H$ with at least $x_0$ edges, where $x_0$ is a witness of the inequality.

Given an integer $k$ and a graph $H$ that contains at least one block with at least $k$ edges, a \emph{$k$-edges leaf block cut} is a tuple $(X,Y,B,v)$, where

\begin{itemize}
\item $B$ is a block with at least $k$ edges,
\item $v\in V(B)$ is a cut vertex, or, in case $H$ is $2$-connected, any vertex of $H=B$,
\item $X$ and $Y$ are two subsets of $V(H)$ such that $X \cup Y = V(H)$ and $X \cap Y = \{v\}$,
\item $Y \setminus \{v\}$ is the vertex set of the connected component of $H \gm \{v\}$ that contains $V(B) \sm \{v\}$, and
\item  $H[Y]$ contains only one block with at least $k$ edges (which is precisely $B$).
\end{itemize}

Intuitively, $B$ would be a leaf of the block-cut tree of the graph $H'$  obtained from  $H$ by  iteratively removing every leaf block with  at most  $k-1$ edges from $\bct(H)$, $v$ is the unique remaining neighbor of $B$ in $\bct(H')$,  $Y$ consists of $B$, and $X$ is the rest of the remaining graph together with $v$.
Note that as long as $H$ contains at least one block with at least $k$ edges, there exists a $k$-edges leaf block cut.

We are now ready to prove Theorem~\ref{hardgene}.

\begin{proof}[Proof of Theorem~\ref{hardgene}]
  Let $\Fcal \subseteq \Ccal$ and let $(F,\ell)$ be the $\Fcal$-M-framework (resp. $\Fcal$-TM-framework) of an input $(G,k)$ of {\sc $k\times k$ Permutation Independent Set}, where $t_\Fcal = 0$.
  Let $H\in \Fcal$ be  {a} graph that minimizes $\besf_H$ with regard to the relation $\prec$ over all the graphs of $\Fcal$.
  Let $(X,Y,B,v)$ be a $5$-edge leaf block cut of $H$.
  Let $b = |{E}(B)|$, let $H_X = H[X]$, and let $H_Y = H[Y]$.
  Let $v'$ be a neighbor of $v$ in $B$ and 
  $H^-_Y$ be the graph obtained from $H_Y$ by removing $\{v, v'\}$.
  Note that $H_X$ and $H_Y^-$ are connected and $H_Y$ contains only one block with at least five edges, and this block is precisely $B$.

  We are now ready to describe the graph $F_\Fcal$. All the new vertices are $J$-extra vertices. Namely,
  starting from $F$, for each $e \in E(G)$, we add a copy of $H_X$, and we denote by $q^e$ the copy of $v$.
  For each $e \in E(G)$ and for each $i \in \intv{1,k}$, we add a copy of $H_Y^-$ where we identify $v$ and $q^e$, and $u$ and $r_i^e$.
  This completes the definition of $F_\Fcal$. We stress that in this construction there are no $B$-extra vertices, i.e., $t_{\cal F}=0$.

  Let $P$ be a solution of {\sc $k\times k$ Permutation Independent Set} on $(G,k)$.
  Then every connected component of $F_\Fcal \gm S_P$ is either a copy of the graph $K$, which is of size $h-1$ (recall that $h = \min_{H \in \Fcal}|V(H)|$), or
  the graph $Z$ depicted in  Figure~\ref{reallocatedstretch}. We claim that
   $\besf_Z \prec \besf_H$ with witness $b$. Indeed, since $|E(H_{Y}^{-})|<|E(H_{Y})|$, $|E(B)| \geq 5$, and the blocks of $Z$ that are not copies of $H_{X}$ or $H_{Y}^{-}$
   have four edges,  it follows that $\besf_Z(b)<\besf_{H}(b)$, and $\besf_Z(b')\leq \besf_{H}(b')$ for all $b'>b$. Therefore,  $\besf_Z \prec \besf_H$ with witness $b$.
   This in turn implies, because of the choice of $H$, that $\besf_Z \prec \besf_H'$ for each $H' \in \Fcal$. Therefore,  no $H'\in \Fcal$ is  a minor of the graph $Z$. Thus  $S_P$ is a solution of  \textsc{$\Fcal$-M-Deletion} (resp. \textsc{$\Fcal$-TM-Deletion}) of size $\ell$.

    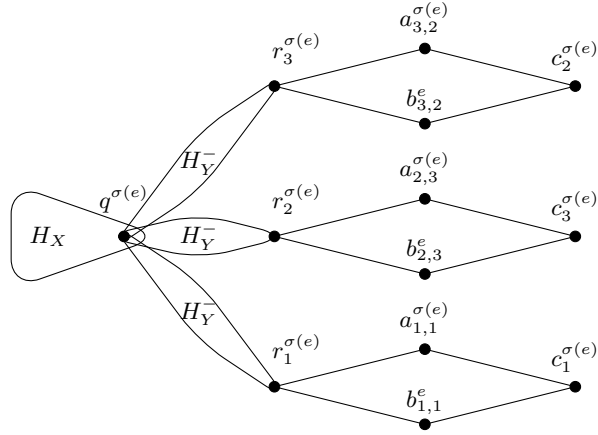
\begin{figure}[htb]
    \centering
    \scalebox{1}{\begin{tikzpicture}[scale=1]
        \draw[white] (0,2) -- (-2,4)  node[pos=0.5]{\textcolor{black}{\footnotesize{$H_Y^-$}}};
        \draw[white] (0,4) -- (-2,4)  node[pos=0.5]{\textcolor{black}{\footnotesize{$H_Y^-$}}};
        \draw[white] (0,6) -- (-2,4)  node[pos=0.5]{\textcolor{black}{\footnotesize{$H_Y^-$}}};

        \node at (0.3,2.5) {\footnotesize{$r^{\sigma(e)}_{1}$}};
        \node at (0.3,4.5) {\footnotesize{$r^{\sigma(e)}_{2}$}};
        \node at (0.3,6.5) {\footnotesize{$r^{\sigma(e)}_{3}$}};

        \node at (2,6.9) {\footnotesize{$a^{\sigma(e)}_{3,2}$}};
        \node at (2,5.8) {\footnotesize{$b^{e}_{3,2}$}};

        \node at (2,4.9) {\footnotesize{$a^{\sigma(e)}_{2,3}$}};
        \node at (2,3.8) {\footnotesize{$b^{e}_{2,3}$}};
        \node at (2,2.9) {\footnotesize{$a^{\sigma(e)}_{1,1}$}};
        \node at (2,1.8) {\footnotesize{$b^{e}_{1,1}$}};

        \node at (4,2.4) {\footnotesize{$c^{\sigma(e)}_{1}$}};
        \node at (4,4.4) {\footnotesize{$c^{\sigma(e)}_{3}$}};
        \node at (4,6.4) {\footnotesize{$c^{\sigma(e)}_{2}$}};

        \node at (-2,4.5) {\footnotesize{$q^{\sigma(e)}$}};

        \foreach \x in {2,4,6}
        {
          \vertex{2,\x+0.5};
          \vertex{2,\x-0.5};
          \vertex{4,\x};
          \draw (0,\x) -- (2, \x+0.5) -- (4,\x) -- (2, \x-0.5) -- (0,\x);
        }

        \vertex{0,2};
        \vertex{0,4};
        \vertex{0,6};

        \vertex{-2,4};


        \draw[rounded corners=4mm] (-2.3,4) -- (-1,4.3) -- (0.2,4) -- (-1,3.7) -- cycle;
        \draw[rounded corners=4mm] (-2.2,3.8) -- (-1,5.4) -- (0.2,6.2) -- (-1,4.6) -- cycle;
        \draw[rounded corners=4mm] (-2.2,4.2) -- (-1,3.4) -- (0.2,1.8) -- (-1,2.6) -- cycle;

        \draw[rounded corners=4mm] (-1.5,4) -- (-3.5,4.7) -- (-3.5,3.3) -- cycle;

        \node at (-3,4) {\footnotesize{$H_X$}};
      \end{tikzpicture}}
    \caption{A connected component $Z$ of $F_H \gm S$ that is not a copy of $K$, with $T^e = \{(1,1), (2,3), (3,2)\}$.
    }

    \label{reallocatedstretch}
  \end{figure}

  Assume now that $S$ is a solution of  \textsc{$\Fcal$-M-Deletion} (resp. \textsc{$\Fcal$-TM-Deletion}) on $F_{\cal F}$ of size $\ell$.
  Let $e \in E(G)$ and let $i,j \in \intv{1,k}$ such that $b^e_{i,j} \not \in S$.
  Let $i' \in \intv{1,k}$ such that $i \not = i'$.
  If $a^{\sigma(e)}_{i',j} \not \in S$, then, since by Lemma~\ref{reimitation} it holds that $S \cap V(J^{\sigma(e)}) = \es$, we have that
  the copy of $H_X$, the copy of $H_Y^-$ between $q^{\sigma(e)}$ and $r^{\sigma(e)}_{i'}$, and
  the path that starts at $q^{\sigma(e)}$, goes through the corresponding copy of $H_Y^-$  until $r^{\sigma(e)}_{i}$, and
  continues with the vertices $b^e_{i,j}, c^{\sigma(e)}_{j}, a^{\sigma(e)}_{i',j}$, and $r^{\sigma(e)}_{i'}$, induce a graph that contains $H$ as a topological minor.
  This implies that if  $a^{\sigma(e)}_{i',j} \not \in S$, $H$ is a topological minor of $F_\Fcal\gm S$.
  As this is not possible by definition of $S$, we have that  $a^{\sigma(e)}_{i',j} \in S$.
  Thus Property~\ref{new-wealthierm} (resp. Property~\ref{new-wealthiertm}) holds and the theorem follows from Lemma~\ref{lemmasol}.
\end{proof}

\subsection{The reduction for \textsc{$\{H\}$-Deletion}}

This section is dedicated to the proofs of Theorem~\ref{hardminor} and Theorem~\ref{hardtminor}, which we restate here for better readability.

\begin{T4}  Let $H \in \mathcal{Q}$.
  Unless the \ETH fails,  \textsc{$\{H\}$-M-Deletion} cannot be solved in time $2^{\smallo(\tw \log \tw)}\cdot n^{\Ocal(1)}$.
\end{T4}

\begin{T5}
  Let $H \in \mathcal{Q}\setminus \Scal$.
  Unless the \ETH fails,  \textsc{$\{H\}$-TM-Deletion} cannot be solved in time $2^{\smallo(\tw \log \tw)}\cdot n^{\Ocal(1)}$.
\end{T5}

Thus, we will focus on cases where the family $\Fcal$ contains only one graph $H$. We start with a number of lemmas, namely Lemma~\ref{cycletwocut} up to Lemma~\ref{crickbisminor}, in which we distinguish several cases according to properties of $H$ such that its  number of cut vertices and the presence of certain cycles and vertices of degree one. Altogether, these cases will cover all the possible graphs $H$ considered in
Theorem~\ref{hardminor} and Theorem~\ref{hardtminor}. The proofs of each of these lemmas are quite similar and follow the same structure. Namely, we first describe the graph $F_{\{H\}}$, and then we prove the equivalence between the existence of solutions of
{\sc $k\times k$ Permutation Independent Set} and \textsc{$\{H\}$-M-Deletion} (or \textsc{$\{H\}$-TM-Deletion}). In the reverse direction, we will prove that Property~\ref{new-wealthierm} and Property~\ref{new-wealthiertm} hold, and therefore we can apply  Lemma~\ref{lemmasol}.

Thanks to Theorem~\ref{hardgene}, we can assume that each block of $H$ contains at most four edges, i.e., each block of $H$ is an edge, a $C_3$, or a $C_4$. In this setting, we have $h = |V(H)|$.


\begin{lemma}
  \label{cycletwocut}
  Let $H$ be a connected graph such that the number of cycles (of size three or four)  in $H$ with at least
  two cut vertices is exactly one.
  Neither \textsc{$\{H\}$-M-Deletion} nor \textsc{$\{H\}$-TM-Deletion} can be solved in time $2^{\smallo(\tw \log \tw)}\cdot n^{\Ocal(1)}$  unless the \ETH fails.
\end{lemma}

\begin{proof}
  Let $(F,\ell)$ be the $\{H\}$-M-framework (resp. $\{H\}$-TM-framework) of an input $(G,k)$ of {\sc $k\times k$ Permutation Independent Set}, where $t_\Fcal = 0$.
  Let $B$ be the block of $H$ with at least three edges and two cut vertices, and let $\{v,v'\}$ be an edge of $B$.
  Let $H^-$ be the graph $H$ where the edge $\{v,v'\}$ has been removed.

  We are now ready to describe the graph $F_{\{H\}}$.
  Starting from $F$, for each $e \in E(G)$ , we introduce a new vertex $q^e$.
  For each $e \in E(G)$ and each $i \in \intv{1,k}$,
  we add a copy of $H^-$ where we identify $v$ and $r_i^e$, and $v'$ and $q^e$.
  This completes the definition of $F_{\{H\}}$.

  Let $P$ be a solution of {\sc $k\times k$ Permutation Independent Set} on $(G,k)$.
  Then every connected component of $F_{\{H\}} \gm S_P$ is either a copy of the graph $K$, which is of size $h-1$, or
  the graph $Z$ depicted in  Figure~\ref{reallocatedbiss}.
  As $Z$ does not contain any cycle that contains at least two cut vertices,  we obtain that $H$ is not a minor of $Z$. Thus  $S_P$ is a solution of  \textsc{$\{H\}$-M-Deletion} (resp. \textsc{$\{H\}$-TM-Deletion}) of size $\ell$.

  \begin{figure}[htb]
    \centering
    \scalebox{1}{\begin{tikzpicture}[scale=1]
        \node at (0,3.5) {\footnotesize{$r^{\sigma(e)}_{1}$}};
        \node at (0,5.5) {\footnotesize{$r^{\sigma(e)}_{2}$}};
        \node at (0,7.5) {\footnotesize{$r^{\sigma(e)}_{3}$}};
        \node at (2,7.9) {\footnotesize{$a^{\sigma(e)}_{3,2}$}};
        \node at (2,6.8) {\footnotesize{$b^{e}_{3,2}$}};

        \node at (2,5.9) {\footnotesize{$a^{\sigma(e)}_{2,3}$}};
        \node at (2,4.8) {\footnotesize{$b^{e}_{2,3}$}};
        \node at (2,3.9) {\footnotesize{$a^{\sigma(e)}_{1,1}$}};
        \node at (2,2.8) {\footnotesize{$b^{e}_{1,1}$}};
        \node at (4,3.4) {\footnotesize{$c^{\sigma(e)}_{1}$}};
        \node at (4,5.4) {\footnotesize{$c^{\sigma(e)}_{3}$}};
        \node at (4,7.4) {\footnotesize{$c^{\sigma(e)}_{2}$}};
        \node at (-3.5,5) {\footnotesize{$q^{\sigma(e)}$}};

        \foreach \x in {3,5,7}
        {
          \vertex{2,\x+0.5};
          \vertex{2,\x-0.5};
          \vertex{4,\x};
          \draw (0,\x) -- (2, \x+0.5) -- (4,\x) -- (2, \x-0.5) -- (0,\x);
          \draw (-3,5) -- (0,\x) node[pos=0.6, circle, draw, fill=white]{\footnotesize{$H^-$}};
        }
        \vertex{0,3};
        \vertex{0,5};
        \vertex{0,7};
        \vertex{-3,5};
      \end{tikzpicture}}
    \caption{A connected component $Z$ of $F_{\{H\}} \gm S_P$ that is not a copy of $K$, with $T^e = \{(1,1), (2,3), (3,2)\}$.}

    \label{reallocatedbiss}
  \end{figure}

    \begin{figure}[htb]
    \centering
    \scalebox{1}{\begin{tikzpicture}[scale=1]
        \node at (0.3,3.5) {\footnotesize{$r^{\sigma(e)}_{1}$}};
        \node at (0.3,5.5) {\footnotesize{$r^{\sigma(e)}_{2}$}};
        \node at (0.3,7.5) {\footnotesize{$r^{\sigma(e)}_{3}$}};
        \node at (2,7.9) {\footnotesize{$a^{\sigma(e)}_{3,2}$}};
        \node at (2,6.8) {\footnotesize{$b^{e}_{3,2}$}};

        \node at (2,5.9) {\footnotesize{$a^{\sigma(e)}_{2,3}$}};
        \node at (2,4.8) {\footnotesize{$b^{e}_{2,3}$}};
        \node at (2,3.9) {\footnotesize{$a^{\sigma(e)}_{1,1}$}};
        \node at (2,2.8) {\footnotesize{$b^{e}_{1,1}$}};
        \node at (4,3.4) {\footnotesize{$c^{\sigma(e)}_{1}$}};
        \node at (4,5.4) {\footnotesize{$c^{\sigma(e)}_{3}$}};
        \node at (4,7.4) {\footnotesize{$c^{\sigma(e)}_{2}$}};
        \node at (-2.5,5) {\footnotesize{$q^{\sigma(e)}$}};

        \draw[white] (0,3) -- (-2,5)  node[pos=0.5]{\textcolor{black}{\footnotesize{$H^-$}}};
        \draw[white] (0,5) -- (-2,5)  node[pos=0.5]{\textcolor{black}{\footnotesize{$H^-$}}};
        \draw[white] (0,7) -- (-2,5)  node[pos=0.5]{\textcolor{black}{\footnotesize{$H^-$}}};
        \draw[rounded corners=4mm] (-2.3,5) -- (-1,5.3) -- (0.2,5) -- (-1,4.7) -- cycle;
        \draw[rounded corners=4mm] (-2.2,4.8) -- (-1,6.4) -- (0.2,7.2) -- (-1,5.6) -- cycle;
        \draw[rounded corners=4mm] (-2.2,5.2) -- (-1,4.4) -- (0.2,2.8) -- (-1,3.6) -- cycle;

        \foreach \x in {3,5,7}
        {
          \vertex{2,\x+0.5};
          \vertex{2,\x-0.5};
          \vertex{4,\x};
          \draw (0,\x) -- (2, \x+0.5) -- (4,\x) -- (2, \x-0.5) -- (0,\x);
        }
        \vertex{0,3};
        \vertex{0,5};
        \vertex{0,7};
        \vertex{-2,5};
      \end{tikzpicture}}
    \caption{A connected component $Z$ of $F_{\{H\}} \gm S_P$ that is not a copy of $K$, with $T^e = \{(1,1), (2,3), (3,2)\}$.}

    \label{reallocatedbissstretch}
  \end{figure}

  Assume now that $S$ is a solution of \textsc{$\{H\}$-M-Deletion} (resp. \textsc{$\{H\}$-TM-Deletion}) on $F_{\{H\}}$ of size $\ell$.
  Let $e \in E(G)$ and let $i,j \in \intv{1,k}$ such that $b^e_{i,j} \not \in S$.
  Let $i' \in \intv{1,k}$ such that $i \not = i'$.
  If $a^{\sigma(e)}_{i',j} \not \in S$, then, since by Lemma~\ref{reimitation} it holds that $S \cap V(J^{\sigma(e)}) = \es$, we have that
    the path $ r^{\sigma(e)}_{i}, b^e_{i,j}, c^{\sigma(e)}_{j}, a^{\sigma(e)}_{i',j}, r^{\sigma(e)}_{i'}$ together with the two copies of $H^-$  attached to $q^{\sigma(e)}$ and $r^{\sigma(e)}_{i}$ and to  $q^{\sigma(e)}$ and $r^{\sigma(e)}_{i'}$  induce a graph that contains $H$ as a topological minor.
  This implies that if  $a^{\sigma(e)}_{i',j} \not \in S$, $H$ is a topological minor of $F_{\{H\}}\gm S$.
  As this is not possible by definition of $S$, we have that  $a^{\sigma(e)}_{i',j} \in S$.
  Thus Property~\ref{new-wealthierm} (resp. Property~\ref{new-wealthiertm}) holds and the lemma follows from Lemma~\ref{lemmasol}.
\end{proof}

We now assume that $H$ contains at least three cut vertices. In particular, this applies to the case where $H=P_{5}$.

\begin{lemma}
  \label{threecut}
  Let $H$ be a connected graph that contains at least three cut vertices.
  Neither \textsc{$\{H\}$-M-Deletion} nor \textsc{$\{H\}$-TM-Deletion} can be solved in time $2^{\smallo(\tw \log \tw)}\cdot n^{\Ocal(1)}$  unless the \ETH fails.
\end{lemma}

\begin{proof}
   Let $(F,\ell)$ be the $\{H\}$-M-framework (resp. $\{H\}$-TM-framework) of an input $(G,k)$ of {\sc $k\times k$ Permutation Independent Set}, where $t_\Fcal$ will be precised later.
  By Lemma~\ref{cycletwocut} and the fact that $H$ has at least three cut vertices, we can assume that $H$ contains
  at least three cut vertices that do not belong to the same block. Therefore, we can find three cut vertices $a$, $c$, $b$ and four blocks $B_a$, $B_{a,c}$, $B_{c,b}$, $B_b$ such that $B_a$ is a leaf of the block-cut tree of $H$ and $B_a$, $a$, $B_{a,c}$, $c$, $B_{c,b}$, $b$, $B_b$ is a path in this block-cut tree.
  Let $a'$ be a vertex of $V(B_a) \sm \{a\}$ and  $r$ be a vertex of $V(B_b) \sm \{b\}$.
  We define $R_a$ to be the connected component of $H \gm \{a',c\}$ that contains $a$,
  $R_b$ to be the connected component of $H \gm \{c,r\}$ that contains $b$,
  $R_c$ to be the connected component of $H \gm (R_a \cup R_b)$ that contains $c$, and
  $R_r$ to be the connected component of $H \gm R_b$ that contains $r$.
  Note that $\{a'\}$, $V(R_a)$, $V(R_c)$, $V(R_b)$, and $V(R_r)$ form a partition of $V(H)$.
  This decomposition of $H$ is depicted in Figure~\ref{decomph}.

  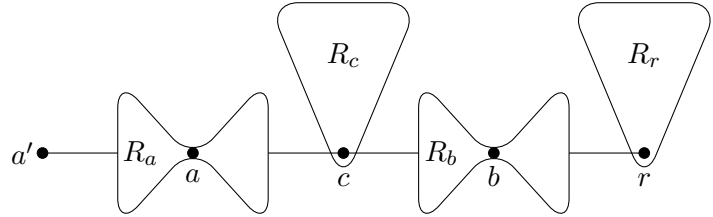
\begin{figure}[htb]
    \centering
    \scalebox{1}{
      \begin{tikzpicture}
        \draw (-1,1) -- (1,1);
        \draw (3,1) -- (4,1);
    \draw (-4,1) -- (-3, 1);
    \draw[rounded corners=4mm] (0,0.6) -- (1,3) -- (-1,3) -- cycle;
    \draw[rounded corners=4mm] (4,0.6) -- (5,3) -- (3,3) -- cycle;
    \vertex{0,1};
    \vertex{4,1};

    \node at (1.3,1) {{$R_b$}};
    \node at (-2.7,1) {{$R_a$}};
    \node at (0,2.3) {{$R_c$}};
    \node at (4,2.3) {{$R_r$}};

    \node at (0,0.65) {{$c$}};
    \node at (4,0.65) {{$r$}};

    \vertex{-2,1};
    \vertex{2,1};
    \node at (-2,0.7) {{$a$}};
    \node at (2,0.7) {{$b$}};

    \vertex{-4,1};
    \node at (-4.2,1) {{$a'~$}};

    \draw[rounded corners=4mm] (1,2) -- (1,0) -- (2,1.1) -- (3,0) -- (3,2) -- (2,0.9) -- cycle;
    \draw[rounded corners=4mm] (-3,2) -- (-3,0) -- (-2,1.1) -- (-1,0) -- (-1,2) -- (-2,0.9) -- cycle;

  \end{tikzpicture}

}
\caption{The decomposition of the graph $H$ where $a$, $c$, and $b$ are three cut vertices.}
\label{decomph}
\end{figure}

We are now ready to describe the graph $F_{\{H\}}$.
Starting from $F$, for each $e \in E(G)$ and each $i \in \intv{1,k}$, we add a copy of $R_c$ where we identify $c$ and $c^e_i$, and a copy of $R_r$ where we identify $r$ and $r^e_i$.
  Moreover, for each $e \in E(G)$ and each $i,j \in \intv{1,k}$ we add a copy of $R_a$ where we identify $a$ and $a^e_{i,j}$ and we connect the vertices $N_H(a') \cap R_a$ to $r^e_i$
  and the vertices $N_H(c) \cap R_a$ to $c^e_j$.
  We also add a copy of $R_b$ where we identify $b$ and $b^e_{i,j}$
  and we connect the vertices $N_H(c) \cap R_b$ to $c^{\sigma(e)}_j$
  and the vertices $N_H(r) \cap R_b$ to $r^{\sigma(e)}_i$.
  Note that the vertices of the copies of $R_a$ and $R_b$ are $B$-extra vertices, and the vertices of the copies of $R_c$ and $R_r$ are $J$-extra vertices.
  In particular, we have $t_\Fcal = |V(R_a) \cup V(R_b)| -2$ and, for each $e \in E(G)$,  $|V(J^e)| = |V(R_c) \cup V(R_r)|\cdot k$.
    This completes the definition of $F_{\{H\}}$.

  Let $P$ be a solution of {\sc $k\times k$ Permutation Independent Set} on $(G,k)$.
  Then every connected component of $F_{\{H\}} \gm S_P$ is either a copy of the graph $K$, which is of size $h-1$, or
  the graph $Z$ depicted in  Figure~\ref{reallocatedbis}.
  As $Z$ contains $h-1$ vertices (both vertices $a'$ and $r$ of $H$ are mapped to $r^e_{i}$), we obtain that $H$ is not a minor of $F_{\{H\}} \gm S_P$.
  Thus  $S_P$ is a solution of  \textsc{$\{H\}$-M-Deletion} (resp. \textsc{$\{H\}$-TM-Deletion}) of size $\ell$.


  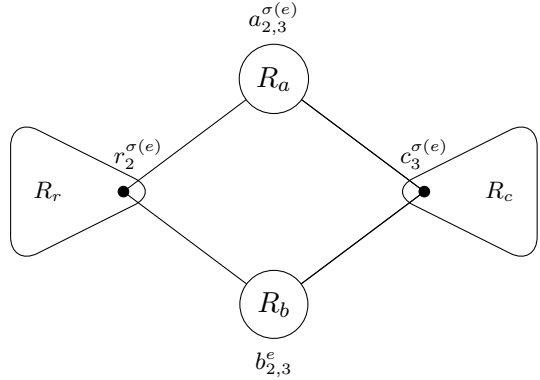
\begin{figure}[htb]
    \centering
    \scalebox{1}{\begin{tikzpicture}[scale=1]
        \node at (0.2,0.5) {\footnotesize{$r^{\sigma(e)}_{2}$}};
        \node at (2,2.3) {\footnotesize{$a^{\sigma(e)}_{2,3}$}};
        \node at (2,-2.3) {\footnotesize{$b^{e}_{2,3}$}};
        \node at (4,0.5) {\footnotesize{$c^{\sigma(e)}_{3}$}};

        \vertex{2,1.5};
        \vertex{2,-1.5};
        \vertex{4,0};
        \draw (0,0) -- (2, 1.5) -- (4,0) -- (2, -1.5) -- (0,0);

        \vertex{0,0};
        \draw[rounded corners=4mm] (0.5,0) -- (-1.5,1) -- (-1.5,-1) -- cycle;
        \node at (-1,0){\footnotesize{$R_r$}};
        \draw[rounded corners=4mm] (3.5,0) -- (5.5,1) -- (5.5,-1) -- cycle;
        \node at (5,0){\footnotesize{$R_c$}};
        \draw (2,1.5) -- (4,0)  node[pos=0, circle, draw, fill=white]{{$R_a$}};
        \draw (2,-1.5) -- (4,0)  node[pos=0, circle, draw, fill=white]{{$R_b$}};

      \end{tikzpicture}}
    \caption{A connected component $Z$ of $F_{\{H\}} \gm S$ that contains $b_{2,3}^e$ with  $(2,3) \in T^e$.}
    \label{reallocatedbis}
  \end{figure}

  Assume now that $S$ is a solution of \textsc{$\{H\}$-M-Deletion} (resp. \textsc{$\{H\}$-TM-Deletion}) on $F_{\{H\}}$ of size $\ell$.
  Let $e \in E(G)$ and let $i,j \in \intv{1,k}$ such that $b^e_{i,j} \not \in S$.
  Let $i' \in \intv{1,k}$ such that $i \not = i'$.
  If $a^{\sigma(e)}_{i',j} \not \in S$, then, since by Lemma~\ref{reimitation} it holds that $S \cap V(J^{\sigma(e)}) = \es$, we have that
  the vertex $r^{\sigma(e)}_{i'}$, the copy of $R_a$ attached to $a^{\sigma(e)}_{i',j}$, the copy of $R_c$ attached to $c^{\sigma(e)}_{j}$, the copy of $R_b$ attached to $ b^e_{i,j}$, and the copy of $R_r$ attached to $r^{\sigma(e)}_{i}$ induce the graph $H$.
  This implies that if  $a^{\sigma(e)}_{i',j} \not \in S$, $H$ is a subgraph of $F_{\{H\}}\gm S$.
  As this is not possible by definition of $S$, we have that  $a^{\sigma(e)}_{i',j} \in S$.
  Thus Property~\ref{new-wealthierm} (resp. Property~\ref{new-wealthiertm}) holds and the lemma follows from Lemma~\ref{lemmasol}.
\end{proof}


In the next lemma, we consider the case where $H$ is a particular type of tree that covers the case where  $H=K_{1,4}$
for the  \textsc{$\{H\}$-M-Deletion} problem.

\begin{lemma}
  \label{twocut}
  Let $H$ be a tree with at most two cut vertices and at least four vertices of degree one.
  \textsc{$\{H\}$-M-Deletion} cannot be solved in time $2^{\smallo(\tw \log \tw)}\cdot n^{\Ocal(1)}$  unless the \ETH fails.
\end{lemma}

\begin{proof}
  Let $(F,\ell)$ be the $\{H\}$-M-framework of an input $(G,k)$ of {\sc $k\times k$ Permutation Independent Set}, where $t_\Fcal$  will be specified later.
  If $H$ has two cut vertices $x$ and $y$, then we set $s_x$ (resp. $s_y$) to be the number of vertices pendent to $x$ (resp. $y$).
  If $H$ has only one cut vertex, we set $s_x = p-2$ and $s_y = 2$, where $p$ is the number of vertices of degree one.

  We are now ready to describe the graph $F_{\{H\}}$.
  Starting from $F$,  for each $e\in E(G)$ and each $i,j \in \intv{1,k}$, we add $s_x-1$ (resp. $s_y-1$) pendent vertices to  $a^e_{i,j}$ (resp. $b^e_{i,j}$).
  Note that the pendent vertices are $B$-extra vertices. 
  In particular we have $t_\Fcal = p-2$ and, for each $e \in E(G)$,  $|V(J^e)| = 2k$, i.e., there are no $J$-extra vertices.
    This completes the definition of $F_{\{H\}}$.

  Let $P$ be a solution of {\sc $k\times k$ Permutation Independent Set} on $(G,k)$.
  Then every connected component of $F_{\{H\}} \gm S_T$ is either a copy of the graph $K$, which is of size $h-1$, or
  the subgraph induced by  $a^e_{i,j}$, $r^e_i$, $b^{\sigma^{-1}(e)}_{i,j}$, $c_j^e$, and the vertices that are pendent to $a^e_{i,j}$ and $b^{\sigma^{-1}(e)}_{i,j}$,  for every $(i,j) \in T$. It can be easily verified that, as by hypothesis,  $s_x+s_y\geq 4$,  this latter subgraph, depicted in Figure~\ref{reoffending}, does not contain $H$ as a  minor. Thus  $F_{\{H\}} \gm S_T$ does not contain $H$ as a  minor and $S_P$ is a solution of  \textsc{$\{H\}$-M-Deletion} of size $\ell$.

  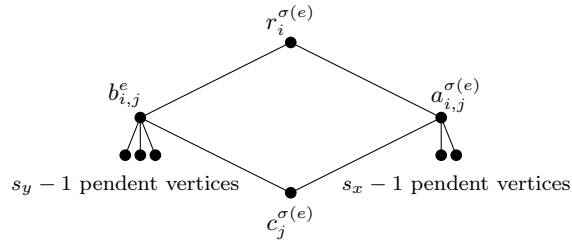
\begin{figure}[htb]
    \centering
    \scalebox{1}{\begin{tikzpicture}[scale=1]
        \node at (0,-0.4) {\footnotesize{$c^{\sigma(e)}_{j}$}};
        \node at (0,2.3) {\footnotesize{$r^{\sigma(e)}_{i}$}};
        \node at (-2.2,1.3) {\footnotesize{$b^e_{i,j}$}};
        \node at (2.2,1.3) {\footnotesize{$a^{\sigma(e)}_{i,j}$}};

        \vertex{0,0};
        \vertex{0,2};
        \vertex{-2,1};
        \vertex{2,1};
        \vertex{2,0.5};
        \draw (2,1) -- (2,0.5);
        \draw (0,0) -- (-2,1) -- (0,2);
        \draw (0,0) -- (2,1) -- (0,2);

        \vertex{2.2,0.5};
        \vertex{-2.2,0.5};
        \vertex{-2,0.5};
        \vertex{-1.8,0.5};
        \draw (-2,1) -- (-2,0.5);
        \draw (-2,1) -- (-1.8,0.5);
        \draw (-2,1) -- (-2.2,0.5);
        \draw (2,1) -- (2.2,0.5);

        \node at (-2.2,0.1) {\scriptsize{$s_y-1$} pendent vertices};
        \node at (2.2,0.1) {\scriptsize{$s_x-1$} pendent vertices};

      \end{tikzpicture}}
    \caption{A connected component of $F_{\{H\}} \gm S$ that is not a copy of $K$, with $s_x = 3$ and $s_y = 4$.}

    \label{reoffending}
  \end{figure}

  Assume now that $S$ is a solution of  \textsc{$\{H\}$-M-Deletion} on $F_{\{H\}}$ of size $\ell$.
  Let $e \in E(G)$ and let $i,j \in \intv{1,k}$ such that $b^e_{i,j} \not \in S$.
  Let $i' \in \intv{1,k}$ such that $i \not = i'$.
  If $a^{\sigma(e)}_{i,'j} \not \in S$, then, as by Lemma~\ref{reimitation} it holds that $S \cap V(J^{\sigma(e)}) = \es$, we have that the path $r^{\sigma(e)}_{i'},a^{\sigma(e)}_{i',j},c^{\sigma(e)}_{j},b^e_{i,j},r^e_{i}$ combined with the $s_x-1$ vertices pendent to $a^{\sigma(e)}_{i',j}$ and the $s_y-1$ vertices pendent to $b^e_{i,j}$ induce a graph $Z$ that contains $H$ as a minor.
  As, by definition of $S$, $F_{\{H\}} \gm S$ does not contain $H$ as a minor, we have that  $a^{\sigma(e)}_{i',j} \in S$.
  Thus Property~\ref{new-wealthierm} holds and the lemma follows from Lemma~\ref{lemmasol}.
\end{proof}

Observe that in the end of the above proof, if $H$ contains two cut vertices, then $Z$ also contains $H$ as a topological minor, but this is not true if $H$ is a star; this is consistent with the single-exponential algorithms given in~\cite{monster2}. Therefore, we  obtain the following lemma for  topological minors.
\begin{lemma}
  \label{twocuttm}
  Let $H$ be a tree with exactly two cut vertices and at least four vertices of degree one.
  \textsc{$\{H\}$-TM-Deletion} cannot be solved in time $2^{\smallo(\tw \log \tw)}\cdot n^{\Ocal(1)}$  unless the \ETH fails.
\end{lemma}

\begin{lemma}
  \label{cyclestartwocut}
  Let $H$ be a connected graph that contains exactly two cut vertices and each cut vertex is part of a cycle.
  Neither \textsc{$\{H\}$-M-Deletion} nor \textsc{$\{H\}$-TM-Deletion} can be solved in time $2^{\smallo(\tw \log \tw)}\cdot n^{\Ocal(1)}$  unless the \ETH fails.
\end{lemma}

\begin{proof}
   Let $(F,\ell)$ be the $\{H\}$-M-framework (resp. $\{H\}$-TM-framework) of an input $(G,k)$ of {\sc $k\times k$ Permutation Independent Set}, where $t_\Fcal= 0$.
  Thanks to Lemma~\ref{cycletwocut}, we can assume that the block containing both cut vertices is not a cycle, hence it is an edge.
  Let $v$ and $v'$ be the two cut vertices and let $H^-$ be the graph obtained from $H$ by contracting the edge $\{v,v'\}$. We denote by $w$ the new vertex.

  We are now ready to describe the graph $F_{\{H\}}$. We set $t_{\cal F}=0$.
  Starting from $F$, for each $e \in E(G)$ and each $i \in \intv{1,k}$ we add a copy of $H^-$ where we identify $w$ and $r^e_i$.
  In particular, for each $e \in E(G)$,  $|V(J^e)| = (|V(H^-)|+1)\cdot k$.
      This completes the definition of $F_{\{H\}}$.


  Let $P$ be a solution of {\sc $k\times k$ Permutation Independent Set} on $(G,k)$.
  Then every connected component of $F_{\{H\}} \gm S_P$ is either a copy of the graph $K$, which is of size $h-1$, or
  the graph $Z$ depicted in  Figure~\ref{reallocatedbisss}.
  As $Z$ has only one cut vertex and every block of this graph is a minor of $C_4$, while $H\not \prem C_{4}$, we obtain that $H$ is not a minor of it.
  Thus $S_P$ is a solution of  \textsc{$\{H\}$-M-Deletion} (resp. \textsc{$\{H\}$-TM-Deletion}) of size $\ell$.

  \begin{figure}[htb]
    \centering
    \scalebox{1}{\begin{tikzpicture}[scale=1]

         \node at (0,4.5) {\footnotesize{$r^{\sigma(e)}_{2}$}};


        \node at (2,4.9) {\footnotesize{$a^{\sigma(e)}_{2,3}$}};
       \node at (2,3.8) {\footnotesize{$b^{e}_{2,3}$}};

        \node at (4,4.4) {\footnotesize{$c^{\sigma(e)}_{3}$}};

        \foreach \x in {4}
        {
          \vertex{2,\x+0.5};
          \vertex{2,\x-0.5};
          \vertex{4,\x};
          \draw (0,\x) -- (2, \x+0.5) -- (4,\x) -- (2, \x-0.5) -- (0,\x);

        }

        \vertex{0,4};



        \draw (-1,4) -- (0,4)  node[pos=0, circle, draw, fill=white]{\footnotesize{$B_b$}};

      \end{tikzpicture}}
    \caption{A connected component $Z$ of $F_{\{H\}} \gm S$ that is not a copy of $K$, that contains $b^e_{2,3}$ with $(2,3) \in P$, where $B_b$ means that we have attached $b\geq 2$ cycles to the vertex $r^{\sigma(e)}_{2}$.}

    \label{reallocatedbisss}
  \end{figure}
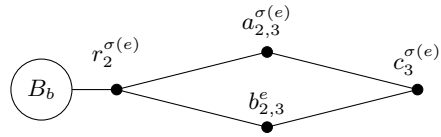

  Assume now that $S$ is a solution of \textsc{$\{H\}$-M-Deletion} (resp.  \textsc{$\{H\}$-TM-Deletion}) on $F_{\{H\}}$ of size $\ell$.
  Let $e \in E(G)$ and let $i,j \in \intv{1,k}$ such that $b^e_{i,j} \not \in S$.
  Let $i' \in \intv{1,k}$ such that $i \not = i'$.
  If $a^{\sigma(e)}_{i',j} \not \in S$, then, since by Lemma~\ref{reimitation}  it holds that $S \cap V(J^{\sigma(e)}) = \es$, we have that
    the path $r^{\sigma(e)}_{i}, b^e_{i,j}, c^{\sigma(e)}_{j}, a^{\sigma(e)}_{i',j}, r^{\sigma(e)}_{i'}$ together with the copies of $H^-$ attached to $r^{\sigma(e)}_{i}$  and $r^{\sigma(e)}_{i'}$ induce a graph that contains $H$ as a minor.
  This implies that if  $a^{\sigma(e)}_{i',j} \not \in S$, $H$ is a topological minor of $F_{\{H\}}\gm S$.
  As this is not possible by definition of $S$, we have that  $a^{\sigma(e)}_{i',j} \in S$.
  Thus Property~\ref{new-wealthierm} (resp. Property~\ref{new-wealthiertm}) holds and the lemma follows from Lemma~\ref{lemmasol}.
\end{proof}

\begin{lemma}\label{commander}
  Let $H$ be a connected graph with exactly two cut vertices such that exactly one of the two cut vertices is part of a cycle.
  Neither \textsc{$\{H\}$-M-Deletion} nor \textsc{$\{H\}$-TM-Deletion}   can be solved in time $2^{\smallo(\tw \log \tw)}\cdot n^{\Ocal(1)}$  unless the \ETH fails.
\end{lemma}

\begin{proof}
  Let $(F,\ell)$ be the $\{H\}$-M-framework (resp. $\{H\}$-TM-framework) of an input $(G,k)$ of {\sc $k\times k$ Permutation Independent Set}, where $t_\Fcal $ will be defined later. By assumption, the block containing both cut vertices is  an edge.
  Let $x$ and $y$ be the two cut vertices.
  Let $C$ be a block that is a cycle, and without loss of generality we may assume that $x \in V(C)$.
  Let $H_x$ (resp. $H_y$) be the connected component of $H \gm ((V(C) \cup \{y\}) \sm \{x\})$ (resp. $H \gm \{x\}$) that contains $x$ (resp. $y$) (see Figure~\ref{rlo4rtp}).

\begin{figure}[htb]

    \centering
    \scalebox{1}{\begin{tikzpicture}[scale=0.9]
\vertex{0,0};
\vertex{2,0};
\vertex{3.5,0.5};
\vertex{3.5,-0.5};
\vertex{3,-1};
\vertex{3,1};
\vertex{-1,0};
\draw (3.5,0.5) -- (2,0) -- (3.5,-0.5);
\draw (3,1) -- (2,0) -- (3,-1);
\draw (-1,0) -- (2,0);
\node at (2,0.3) {{$y$}};
\node at (0.2,0.3) {{$x$}};
\node at (-0.25,1.3) {{$C$}};

\draw[rounded corners=4mm, dashed] (1.7,-0.5) -- (1.7, 0.5) --  (3,1.5) -- (4.1,0.5) -- (4.1,-0.5) -- (3,-1.5) -- cycle;

\draw[rounded corners=1mm] (0,0) -- (0,1) -- (-0.5,1) -- (0,0);
\draw[rounded corners=1mm] (0,0) -- (0,-1) -- (-0.5,-1) -- (0,0);
\draw[rounded corners=1mm] (0,0) -- (-1,-1) -- (-1,-0.5) -- (0,0);

\draw[rounded corners=4mm, dashed] (0.2,0.3) -- (-1.2, 0.3) --  (-1.2,-1.2) -- (0.2,-1.2) -- cycle;

\node at (4.3,-0.7) {{$H_y$}};
\node at (-1.7,-0.7) {{$H_x$}};

\end{tikzpicture}
}

\caption{A visualization of $C$, $H_{x}$, and $H_{y}$ in $H$.}
\label{rlo4rtp}
\end{figure}
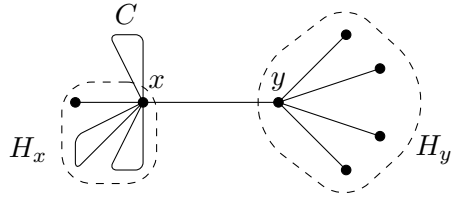

%
%

  We are now ready to describe the graph $F_{\{H\}}$.
  Starting from $F$, we add, for each $e \in E(G)$ and each $i,j \in \intv{1,k}$, a vertex $\overline{a}^e_{i,j}$ and the edges $\{\overline{a}^e_{i,j},r^e_i\}$ and $\{c^e_j,\overline{a}^{e}_{i,j}\}$.
  Moreover, for each  $e \in E(G)$ and each $j \in \intv{1,k}$, we add a copy of $H_x$ where we identify $x$ and $c^e_j$, and a copy of $H_y$ where we identify $y$ and $r^e_{j}$.
  The vertices in the copies of $H_x$ and the copies of $H_y$ are $J$-extra vertices and the vertices $\overline{a}^e_{i,j}$,  $e \in E(G)$ and $i,j \in \intv{1,k}$
  are $B$-extra vertices.
  In particular, we have $t_\Fcal = 1$ and, for each $e \in E(G)$,  $|V(J^e)| = (|V(H_x)|+|V(H_y)|)\cdot k$.
    This completes the definition of $F_{\{H\}}$.

  Let $P$ be a solution of {\sc $k\times k$ Permutation Independent Set} on $(G,k)$.
  The connected components of  $F_{\{H\}} \gm S_P$ are either copies of the graph $K$, which is of size $h-1$, or the graph $Z$ depicted in Figure~\ref{pic1-cobannerbis}.
  Since $H_{y}$ has no cycles and $H_{x}$ has one cycle less than $H$, it follows that if $Z$ contains $H$ as a minor, then there is a cycle of this minor that contains both $c_j^{\sigma(e)}$ and $r_i^{\sigma(e)}$.
  But in that case we cannot find in $Z$ a block consisting of one edge whose both endpoints are cut vertices, corresponding to  the edge $\{x,y\}$.
  We obtain that $H$ is not a minor of the depicted graph.
   Thus  $F_{\{H\}} \gm S_P$ does not contain $H$ as a minor and $S_P$ is a solution of \textsc{$\{H\}$-M-Deletion} (resp. \textsc{$\{H\}$-TM-Deletion}) of size $\ell$.

  \begin{figure}[htb]
    \centering
    \scalebox{1}{\begin{tikzpicture}[scale=0.9]
        \node at (0,4.5) {\footnotesize{$r^{\sigma(e)}_{i}$}};
        \node at (2.4,0.5) {\footnotesize{$c^{\sigma(e)}_{j}$}};
        \node at (1.6,2.6) {\footnotesize{$a^{\sigma(e)}_{i,j}$}};
        \node at (3.5,3) {\footnotesize{$\overline{a}^{\sigma(e)}_{i,j}$}};
        \node at (-2.5,2) {\footnotesize{$b^{e}_{i,j}$}};
        \vertex{0,4};
        \vertex{2,3};
        \vertex{3,3};
        \vertex{2.5,1};
        \vertex{-2,2};
        \draw (0,4) -- (2,3) -- (2.5,1) -- (-2,2) -- (0,4) -- (3,3) -- (2.5,1);


        \draw[rounded corners=4mm] (2,1) -- (4,0.3) -- (4,1.7) -- cycle;
        \node at (3.5,1) {\footnotesize{$H_x$}};
        \draw[rounded corners=4mm] (0.5,4) -- (-1.5,4.7) -- (-1.5,3.3) -- cycle;
        \node at (-1,4) {\footnotesize{$H_y$}};


      \end{tikzpicture}}
    \caption{A connected component $Z$ of $F_{\{H\}} \gm S$ that is not a copy of $K$.
    }

    \label{pic1-cobannerbis}
  \end{figure}
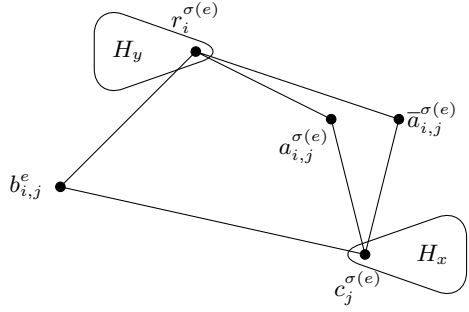
%

  Assume now that $S$ is a solution of \textsc{$\{H\}$-M-Deletion} (resp. \textsc{$\{H\}$-TM-Deletion})
  on $F_{\{H\}}$ of size $\ell$.
  Let $e \in E(G)$ and let $i,j \in \intv{1,k}$ such that $b^e_{i,j} \not \in S$.
  Let $i' \in \intv{1,k}$ such that $i \not = i'$.
  If $a^{\sigma(e)}_{i',j} \not \in S$ then, as by Lemma~\ref{reimitation} $S \cap V(J^{\sigma(e)}) = \es$, it follows that the $C_{4}$ induced by $r^{\sigma(e)}_{i'},a^{\sigma(e)}_{i',j},c^{\sigma(e)}_{j},\overline{a}^{\sigma(e)}_{i',j},$ the path
  $c^{\sigma(e)}_{j},b^e_{i,j},r^{\sigma(e)}_{i}$,
together with the copy of $H_x$ attached to $c^{\sigma(e)}_{j}$ and the copy of $H_y$ attached to $r^{\sigma(e)}_{i}$ induce a subgraph  of $F_{\{H\}}\gm S$ that contains $H$ as a topological minor. This subgraph, together with an extra copy of $H_{y}$ attached to $r_{i'}^{\sigma(e)}$, is depicted in Figure~\ref{pic2-cobannerbis}.
  As this is forbidden by definition of $S$,
  we have that  $a^{\sigma(e)}_{i',j} \in S$.
  Thus Property~\ref{new-wealthierm} (resp. Property~\ref{new-wealthiertm}) holds and the lemma follows from Lemma~\ref{lemmasol}.
\end{proof}

\vspace{-.2cm}
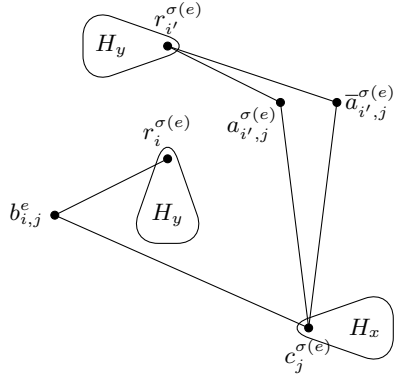
\begin{figure}[htb]
  \centering
  \scalebox{1}{\begin{tikzpicture}[scale=0.75]

      \node at (0,4.5) {\footnotesize{$r^{\sigma(e)}_{i}$}};
      \node at (0.2,6.5) {\footnotesize{$r^{\sigma(e)}_{i'}$}};
      \node at (2.5,0.5) {\footnotesize{$c^{\sigma(e)}_{j}$}};
      \node at (1.5,4.6) {\footnotesize{$a^{\sigma(e)}_{i',j}$}};
      \node at (3.6,5) {\footnotesize{$\overline{a}^{\sigma(e)}_{i',j}$}};
      \node at (-2.5,3) {\footnotesize{$b^{e}_{i,j}$}};

      \vertex{0,4};
      \vertex{0,6};
      \vertex{2,5};
      \vertex{3,5};
      \vertex{2.5,1};
      \vertex{-2,3};

      \draw (0,6) -- (2,5) -- (2.5,1) -- (-2,3) -- (0,4);
      \draw (0,6) -- (3,5) -- (2.5,1);


      \draw[rounded corners=4mm] (2,1) -- (4,0.3) -- (4,1.7) -- cycle;
      \node at (3.5,1) {\footnotesize{$H_x$}};
      \draw[rounded corners=4mm] (0.5,6) -- (-1.5,6.7) -- (-1.5,5.3) -- cycle;
      \node at (-1,6) {\footnotesize{$H_y$}};
      \draw[rounded corners=4mm] (0,4.5) -- (0.7,2.5) -- (-0.7,2.5) -- cycle;
      \node at (0,3) {\footnotesize{$H_y$}};

    \end{tikzpicture}}
  \caption{A connected component $Z$ of $F_{\{H\}} \gm S$ if $b^e_{i,j} \not \in S$ and $a^{\sigma(e)}_{i',j} \not \in S$ for some $e \in E(G)$, $i,i',j \in \intv{1,k}$, $i \not = i'$.}

  \label{pic2-cobannerbis}
\end{figure}

\begin{lemma}\label{butternutbis}
  Let $H$ be a connected graph with exactly one cut vertex and at least two cycles.
  Neither \textsc{$\{H\}$-M-Deletion} nor \textsc{$\{H\}$-TM-Deletion} can be solved in time $2^{\smallo(\tw \log \tw)}\cdot n^{\Ocal(1)}$  unless the \ETH fails.
\end{lemma}

\begin{proof}
     Let $(F,\ell)$ be the $\{H\}$-M-framework (resp. $\{H\}$-TM-framework) of an input $(G,k)$ of {\sc $k\times k$ Permutation Independent Set}, where $t_\Fcal$ will be specified later.
  Let $x$ be the cut vertex of $H$, and let $B_1$ and $B_2$ be two blocks that are cycles.
  We define the graph $H_x$ to be $H \gm (V(B_1 \cup B_2) \sm \{x\})$.

  We are now ready to describe the graph $F_{\{H\}}$.
  Starting from $F$, we add,  for  each $e \in E(G)$ and each $i,j \in \intv{1,k}$, two new vertices $\overline{a}_{i,j}^{e}$ and $\overline{b}_{i,j}^{e}$ and the edges $\{\overline{a}^e_{i,j},r^e_i\}$, $\{r^e_i,\overline{b}^{\sigma^{-1}(e)}_{i,j}\}$, $\{\overline{b}^{\sigma^{-1}(e)}_{i,j},c^e_j\}$, and $\{c^e_j,\overline{a}_{i,j}\}$.
  Then for each $e \in E(G)$ and each $j \in \intv{1,k}$ we add a copy of $H_x$ where we identify $x$ with $c^e_j$.
  The vertices in the copies of $H_x$ are $J$-extra vertices and the vertices $\overline{a}_{i,j}^{e}$ and $\overline{b}_{i,j}^{e}$, $i,j \in \intv{1,k}$ and $e \in E(G)$, are $B$-extra vertices.
    In particular we have $t_\Fcal = 2$ and, for each $e \in E(G)$,  $|V(J^e)| = (|V(H_x)|+1)\cdot k$.
    This completes the definition of $F_{\{H\}}$.

    Let $P$ be a solution of {\sc $k\times k$ Permutation Independent Set} on $(G,k)$.
      Then every connected component of $F_{\{H\}} \gm S_P$ is either a copy of the graph $K$, which is of size $h-1$, or
  the graph $Z$ depicted in  Figure~\ref{butbis}.
  Since $Z$ has only one block more than $H_{x}$, it follows that $Z$ does not contain $H$ as a minor.
  Thus  $F \gm S_T$ does not contain $H$ as a minor, and $S_P$ is a solution of  \textsc{$\{H\}$-M-Deletion} (resp. \textsc{$\{H\}$-TM-Deletion}) of size $\ell$.
  \begin{figure}[htb]
    \centering
    \scalebox{1}{\begin{tikzpicture}[scale=0.8]

        \node at (0,4.5) {\footnotesize{$r^{\sigma(e)}_{i}$}};
        \node at (2.2,0.5) {\footnotesize{$c^{\sigma(e)}_{j}$}};
        \node at (1.6,2.6) {\footnotesize{$a^{\sigma(e)}_{i,j}$}};
        \node at (3.6,3) {\footnotesize{$\overline{a}^{\sigma(e)}_{i,j}$}};
        \node at (-3.45,3) {\footnotesize{$b^{e}_{i,j}$}};
        \node at (-1.2,3) {\footnotesize{$\overline{b}^{e}_{i,j}$}};

        \vertex{0,4};
        \vertex{2,3};
        \vertex{3,3};
        \vertex{2.5,1};
        \vertex{-2,3};
        \vertex{-3,3};

        \draw (0,4) -- (2,3) -- (2.5,1) -- (-2,3) -- (0,4) -- (3,3) -- (2.5,1) -- (-3,3) -- (0,4);

      \draw[rounded corners=4mm] (2,1) -- (4,0.3) -- (4,1.7) -- cycle;
      \node at (3.5,1) {\footnotesize{$H_x$}};
      \end{tikzpicture}}
    \caption{A connected component $Z$ of $F_{\{H\}} \gm S$ that is not a copy of $K$.}
    \label{butbis}
  \end{figure}
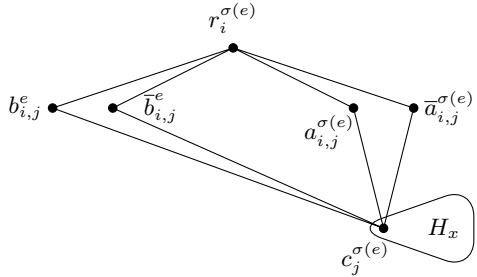

  Assume now that $S$ is a solution of \textsc{$\{H\}$-M-Deletion}
  on $F_{\{H\}}$ of size $\ell$.
  Let $e \in E(G)$ and let $i,j \in \intv{1,k}$ such that $b^e_{i,j} \not \in S$.
  Let $i' \in \intv{1,k}$ such that $i \not = i'$.
  If $a^{\sigma(e)}_{i',j} \not \in S$ then, as by Lemma~\ref{reimitation} $S \cap V(J^{\sigma(e)}) = \es$, we have that the graph
  induced by the two paths $r^{\sigma(e)}_{i'},a^{\sigma(e)}_{i',j},c^{\sigma(e)}_{j},b^e_{i,j},r^{\sigma(e)}_{i}$
  and $r^{\sigma(e)}_{i'},\overline{a}^{\sigma(e)}_{i',j},c^{\sigma(e)}_{j},\overline{b}^e_{i,j},r^{\sigma(e)}_{i}$, together with the copy of $H_x$ attached to $c^{\sigma(e)}_{j}$,
  depicted in Figure~\ref{countdownbis},
  is a subgraph of $F\gm S$ containing $H$ as a topological minor.
  As this is forbidden by definition of $S$,
  we have that  $a^{\sigma(e)}_{i',j} \in S$.
  Thus Property~\ref{new-wealthierm} (resp. Property~\ref{new-wealthiertm}) holds and the lemma follows from Lemma~\ref{lemmasol}.
  \begin{figure}[htb]
    \centering
    \scalebox{1}{\begin{tikzpicture}[scale=0.7]

        \node at (0,4.5) {\footnotesize{$r^{\sigma(e)}_{i}$}};
        \node at (0,6.5) {\footnotesize{$r^{\sigma(e)}_{i'}$}};
        \node at (2.2,0.5) {\footnotesize{$c^{\sigma(e)}_{j}$}};
        \node at (1.5,4.6) {\footnotesize{$a^{\sigma(e)}_{i',j}$}};
        \node at (3.7,5) {\footnotesize{$\overline{a}^{\sigma(e)}_{i',j}$}};
        \node at (-3.5,3) {\footnotesize{$b^{e}_{i,j}$}};
        \node at (-1.,3) {\footnotesize{$\overline{b}^{e}_{i,j}$}};

        \vertex{0,4};
        \vertex{0,6};
        \vertex{2,5};
        \vertex{3,5};
        \vertex{2.5,1};
        \vertex{-2,3};
        \vertex{-3,3};

        \draw (0,6) -- (2,5) -- (2.5,1) -- (-2,3) -- (0,4);
        \draw (0,6) -- (3,5) -- (2.5,1) -- (-3,3) -- (0,4);
        \draw[rounded corners=4mm] (2,1) -- (4,0.3) -- (4,1.7) -- cycle;
      \node at (3.5,1) {\footnotesize{$H_x$}};

      \end{tikzpicture}}
    \caption{A connected component $Z$ of $F_{\{H\}} \gm S$, if $b^e_{i,j} \not \in S$ and $a^{\sigma(e)}_{i',j} \not \in S$ for some $e \in E(G)$, $i,i',j \in \intv{1,k}$, $i \not = i'$.}

    \label{countdownbis}
  \end{figure}
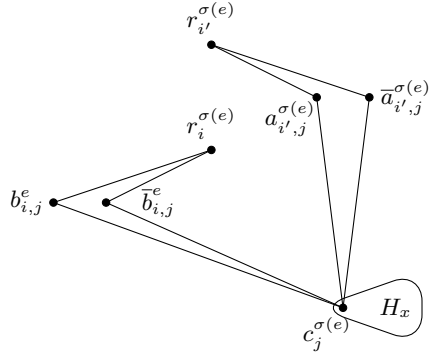
\end{proof}

In the next two lemmas, namely Lemma~\ref{crickbisminor} and Lemma~\ref{crickbistm}, we deal separately with the minor and topological minor versions, respectively.

\begin{lemma}\label{crickbisminor}
  Let $H$ be a connected graph with exactly one cut vertex and exactly one cycle such that $H$ is not a minor of the \banner.
  \textsc{$\{H\}$-M-Deletion} cannot be solved in time $2^{\smallo(\tw \log \tw)}\cdot n^{\Ocal(1)}$  unless the \ETH fails.
\end{lemma}

\begin{proof}
     Let $(F,\ell)$ be the $\{H\}$-M-framework of an input $(G,k)$ of {\sc $k\times k$ Permutation Independent Set}, where $t_\Fcal$ will be specified  later.
  Let $s$ be the number of vertices of degree one in $H$.
  As $H$ is not a minor of the \banner and (because of ~Theorem~\ref{hardgene}) we can assume that each block of $H$ contains at most four edges,
  we have that  $s \geq 2$.

  We are now ready to describe the graph $F_{\{H\}}$.
  Starting from $F$, we add,  for each $e \in E(G)$ and each $j \in \intv{1,k}$, three new vertices $d^e_j$, $f^e_j$, and $g^e_j$ and the edges $\{d^e_j,f^e_j\}$ and $\{f^e_j,g^e_j\}$.
Moreover, for each $e \in E(G)$ and each $i,j \in \intv{1,k}$, we add the edges
$\{b^{\sigma^{-1}(e)}_{i,j},d^e_j\}$ and $\{g^e_j,a_{i,j}\}$, and $s-2$ vertices pendent to $b^e_{i,j}$.
The vertices $d^e_j$, $f^e_j$, and $g^e_j$, $j \in \intv{1,k}$ and $e \in E(G)$, are $J$-extra vertices, and the pendent vertices are $B$-extra vertices.
  In particular we have $t_\Fcal = s-2$ and, for each $e \in E(G)$,  $|V(J^e)| = 5 k$.
    This completes the definition of $F_{\{H\}}$.


    Let $P$ be a solution of {\sc $k\times k$ Permutation Independent Set} on $(G,k)$.
      Then every connected component of $F_{\{H\}} \gm S_P$ is either a copy of the graph $K$, which is of size $h-1$, or
  the graph $Z$ depicted in  Figure~\ref{cowbis}.
  Note that  $Z$ contains three different cycles, but for each of them some
  pendent edge is missing in order to find $H$ as a minor.
   Thus  $F \gm S_P$ does not contain  $H$ as a minor, and $S_P$ is a solution of \textsc{$\{H\}$-M-Deletion}.
  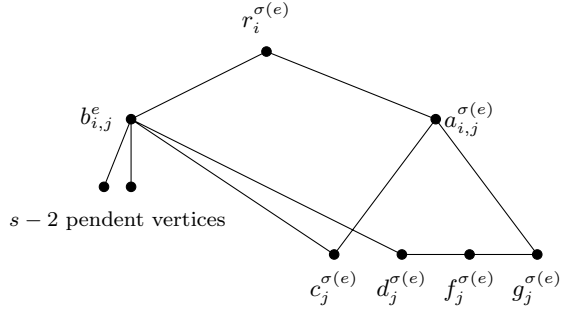
\begin{figure}[htb]
    \centering
    \scalebox{1}{\begin{tikzpicture}[scale=0.9]
        \node at (0,4.5) {\footnotesize{$r^{\sigma(e)}_{i}$}};
        \node at (1,0.5) {\footnotesize{$c^{\sigma(e)}_{j}$}};
        \node at (2,0.5) {\footnotesize{$d^{\sigma(e)}_{j}$}};
        \node at (3,0.5) {\footnotesize{$f^{\sigma(e)}_{j}$}};
        \node at (4,0.5) {\footnotesize{$g^{\sigma(e)}_{j}$}};
        \node at (3,3) {\footnotesize{$a^{\sigma(e)}_{i,j}$}};
        \node at (-2.5,3) {\footnotesize{$b^{e}_{i,j}$}};

        \vertex{0,4};
        \vertex{2.5,3};
        \vertex{1,1};
        \vertex{2,1};
        \vertex{3,1};
        \vertex{4,1};
        \vertex{-2,3};

        \vertex{-2,2};
        \vertex{-2.4,2};
        \draw (-2,3) -- (-2,2);
        \draw (-2,3) -- (-2.4,2);

        \node at (-2.2,1.5) {\scriptsize{$s-2$} pendent vertices};

        \draw (-2,3) -- (1,1) -- (2.5,3) -- (0,4) -- (-2,3) -- (2,1) -- (4,1) -- (2.5,3);

      \end{tikzpicture}}
    \caption{A connected component $Z$ of $F_{\{H\}} \gm S$ that is not a copy of $K$ with $s = 4$.}
    \label{cowbis}
  \end{figure}

  Assume now that $S$ is a solution of \textsc{$\{H\}$-M-Deletion}.
  Let $e \in E(G)$ and let $i,j \in \intv{1,k}$ such that $b^e_{i,j} \not \in S$.
  Let $i' \in \intv{1,k}$ such that $i \not = i'$.
  If $a^{\sigma(e)}_{i',j} \not \in S$ then, as by Lemma~\ref{reimitation}  $S \cap V(J^{\sigma(e)}) = \es$, we have that the graph $Z$ induced by the two paths $r^{\sigma(e)}_{i'},a^{\sigma(e)}_{i',j},c^{\sigma(e)}_{j},b^e_{i,j},r^{\sigma(e)}_{i}$
  and ${a}^{\sigma(e)}_{i',j},g^{\sigma(e)}_{j}, f^{\sigma(e)}_{j}, {d}^{\sigma(e)}_{j}, {b}^e_{i,j}$, and the $s-2$ vertices pendent to ${b}^e_{i,j}$, depicted in Figure~\ref{supercowbis},
  is a subgraph of $F_{\{H\}}\gm S$ containing $H$ as a minor.
  As this is forbidden by definition of $S$,
  we have that  $a^{\sigma(e)}_{i',j} \in S$.
  Thus Property~\ref{new-wealthierm} holds and the lemma follows from Lemma~\ref{lemmasol}.
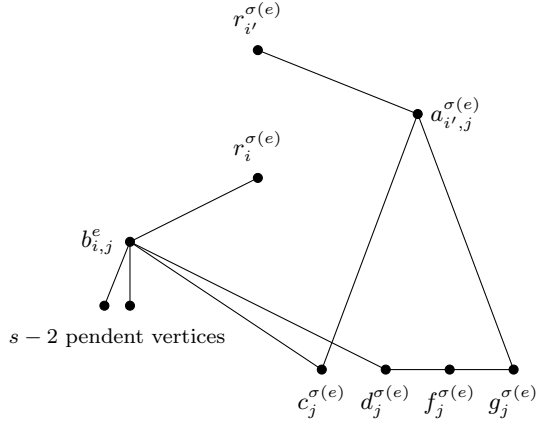
\begin{figure}[htb]
  \centering
  \scalebox{1}{\begin{tikzpicture}[scale=0.85]
      \node at (0,4.5) {\footnotesize{$r^{\sigma(e)}_{i}$}};
      \node at (0,6.5) {\footnotesize{$r^{\sigma(e)}_{i'}$}};
        \node at (1,0.5) {\footnotesize{$c^{\sigma(e)}_{j}$}};
        \node at (2,0.5) {\footnotesize{$d^{\sigma(e)}_{j}$}};
        \node at (3,0.5) {\footnotesize{$f^{\sigma(e)}_{j}$}};
        \node at (4,0.5) {\footnotesize{$g^{\sigma(e)}_{j}$}};
      \node at (3.1,5) {\footnotesize{$a^{\sigma(e)}_{i',j}$}};
      \node at (-2.5,3) {\footnotesize{$b^{e}_{i,j}$}};

      \vertex{0,4};
      \vertex{0,6};
      \vertex{2.5,5};
      \vertex{1,1};
      \vertex{2,1};
      \vertex{3,1};
      \vertex{4,1};
      \vertex{-2,3};

      \draw (-2,3) -- (1,1) -- (2.5,5) -- (0,6);
      \draw (0,4) -- (-2,3) -- (2,1) -- (4,1) -- (2.5,5);
      \vertex{-2,2};
      \vertex{-2.4,2};
      \draw (-2,3) -- (-2,2);
      \draw (-2,3) -- (-2.4,2);

        \node at (-2.2,1.5) {\scriptsize{$s-2$} pendent vertices};

    \end{tikzpicture}}
  \caption{A connected component $Z$ of $F_{\{H\}} \gm S$, if $b^e_{i,j} \not \in S$ and $a^{\sigma(e)}_{i',j} \not \in S$ for some $e \in E(G)$, $i,i',j \in \intv{1,k}$, $i \not = i'$ where $s = 4$.}

  \label{supercowbis}
\end{figure}
\end{proof}

\begin{lemma}\label{crickbistm}
  Let $H$ be a connected graph with exactly one cut vertex and exactly one cycle such that $H$ is not a minor of the \banner.
  \textsc{$\{H\}$-TM-Deletion} cannot be solved in time $2^{\smallo(\tw \log \tw)}\cdot n^{\Ocal(1)}$  unless the \ETH fails.
\end{lemma}

\begin{proof}
     Let $(F,\ell)$ be the $\{H\}$-TM-framework of an input $(G,k)$ of {\sc $k\times k$ Permutation Independent Set}, where $t_\Fcal$ will be specified  later.
  Let $s$ be the number of vertices of degree one in $H$.
  As in Lemma~\ref{crickbisminor}, since $H$ is not a minor of the \banner and (because of ~Theorem~\ref{hardgene}) we can assume that each block of $H$ contains at most four edges,
  we have that  $s \geq 2$.

  We are now ready to describe the graph $F_{\{H\}}$.
  Starting from $F$, we add,  for each $e \in E(G)$, a vertex $q^{e}$ and $s$ vertices pendent to $q^e$, and for each $e \in E(G)$ and each $j \in \intv{1,k}$, the edge $\{q^e, r^e_{j}\}$.
The vertices $q^e$, $e \in E(G)$,  and the pendent vertices are $J$-extra vertices.
  In particular we have $t_\Fcal = 0$ and, for each $e \in E(G)$,  $|V(J^e)| = 2k+s+1$.
  This completes the definition of $F_{\{H\}}$.
  Note that this construction is similar to the construction provided in Section~\ref{globred} with $H_x$ being a star with $s$ leaves, $x$ the non-leaf vertex, and $H_Y^-$ an edge.


    Let $P$ be a solution of {\sc $k\times k$ Permutation Independent Set} on $(G,k)$.
      Then every connected component of $F_{\{H\}} \gm S_P$ is either a copy of the graph $K$, which is of size $h-1$, or
  the graph $Z$ depicted in  Figure~\ref{reallocatedbissstretchnew}.
  Note that each vertex of $Z$ contained in a cycle is of degree at most three. Since $s\geq 2$, there is a vertex in $H$ of degree at least four contained in a cycle.
   Thus  $F \gm S_P$ does not contain  $H$ as a topological minor, and $S_P$ is a solution of \textsc{$\{H\}$-TM-Deletion}.
    \begin{figure}[htb]
    \centering
    \scalebox{1}{\begin{tikzpicture}[scale=1]
        \node at (0.3,3.5) {\footnotesize{$r^{\sigma(e)}_{1}$}};
        \node at (0.3,5.5) {\footnotesize{$r^{\sigma(e)}_{2}$}};
        \node at (0.3,7.5) {\footnotesize{$r^{\sigma(e)}_{3}$}};
        \node at (2,7.9) {\footnotesize{$a^{\sigma(e)}_{3,2}$}};
        \node at (2,6.8) {\footnotesize{$b^{e}_{3,2}$}};

        \node at (2,5.9) {\footnotesize{$a^{\sigma(e)}_{2,3}$}};
        \node at (2,4.8) {\footnotesize{$b^{e}_{2,3}$}};
        \node at (2,3.9) {\footnotesize{$a^{\sigma(e)}_{1,1}$}};
        \node at (2,2.8) {\footnotesize{$b^{e}_{1,1}$}};
        \node at (4,3.4) {\footnotesize{$c^{\sigma(e)}_{1}$}};
        \node at (4,5.4) {\footnotesize{$c^{\sigma(e)}_{3}$}};
        \node at (4,7.4) {\footnotesize{$c^{\sigma(e)}_{2}$}};
        \node at (-2.5,5) {\footnotesize{$q^{\sigma(e)}$}};


        \foreach \x in {3,5,7}
        {
          \vertex{2,\x+0.5};
          \vertex{2,\x-0.5};
          \vertex{4,\x};
          \draw (0,\x) -- (2, \x+0.5) -- (4,\x) -- (2, \x-0.5) -- (0,\x);
          \draw (-2,5) -- (0,\x);
        }
        \vertex{0,3};
        \vertex{0,5};
        \vertex{0,7};
        \vertex{-2,5};

        \vertex{-2,4};
        \vertex{-2.4,4};
        \draw (-2,5) -- (-2,4);
        \draw (-2,5) -- (-2.4,4);

        \node at (-2.5,3.5) {\scriptsize{$s$} pendent vertices};

      \end{tikzpicture}}
    \caption{A connected component $Z$ of $F_{\{H\}} \gm S_P$ that is not a copy of $K$, with $T^e = \{(1,1), (2,3), (3,2)\}$ and $s=2$.}

    \label{reallocatedbissstretchnew}
  \end{figure}
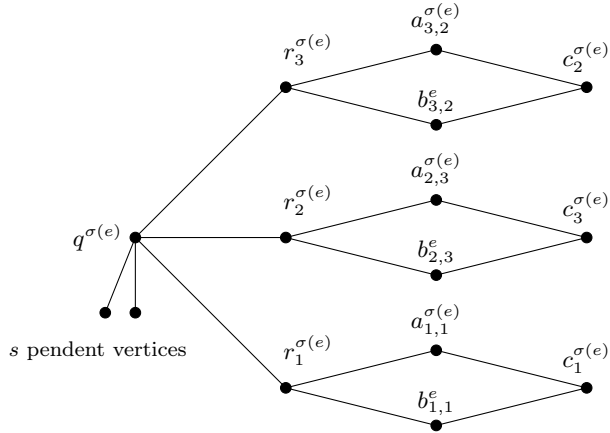

  Assume now that $S$ is a solution of \textsc{$\{H\}$-TM-Deletion}.
  Let $e \in E(G)$ and let $i,j \in \intv{1,k}$ such that $b^e_{i,j} \not \in S$.
  Let $i' \in \intv{1,k}$ such that $i \not = i'$.
  If $a^{\sigma(e)}_{i',j} \not \in S$ then, as by Lemma~\ref{reimitation}  $S \cap V(J^{\sigma(e)}) = \es$, we have that the graph $Z$ induced,
  by the cycle $q^{\sigma(e)},r^{\sigma(e)}_{i'},a^{\sigma(e)}_{i',j},c^{\sigma(e)}_{j},b^e_{i,j},r^{\sigma(e)}_{i}, q^{\sigma(e)}$
  and the $s$ vertices pendent to $q^{\sigma(e)}$,
  is a subgraph of $F_{\{H\}}\gm S$ containing $H$ as a topological minor.
  This situation is depicted in Figure~\ref{reallocatedbissstretchnewbis}.
  As this is forbidden by definition of $S$,
  we have that  $a^{\sigma(e)}_{i',j} \in S$.
  Thus Property~\ref{new-wealthiertm} holds and the lemma follows from Lemma~\ref{lemmasol}.
      \begin{figure}[htb]
    \centering
    \scalebox{1}{\begin{tikzpicture}[scale=1]
        \node at (0.3,3.5) {\footnotesize{$r^{\sigma(e)}_{1}$}};
        \node at (0.3,5.5) {\footnotesize{$r^{\sigma(e)}_{2}$}};
        \node at (0.3,7.5) {\footnotesize{$r^{\sigma(e)}_{3}$}};
        \node at (2,7.9) {\footnotesize{$a^{\sigma(e)}_{3,2}$}};
        \node at (2,6.8) {\footnotesize{$b^{e}_{3,3}$}};

        \node at (2,5.9) {\footnotesize{$a^{\sigma(e)}_{2,3}$}};
        \node at (1.9,4.8) {\footnotesize{$b^{e}_{2,2}$}};
        \node at (2,3.9) {\footnotesize{$a^{\sigma(e)}_{1,1}$}};
        \node at (2,2.8) {\footnotesize{$b^{e}_{1,1}$}};
        \node at (4,3.4) {\footnotesize{$c^{\sigma(e)}_{1}$}};
        \node at (4,5.4) {\footnotesize{$c^{\sigma(e)}_{3}$}};
        \node at (4,7.4) {\footnotesize{$c^{\sigma(e)}_{2}$}};
        \node at (-2.5,5) {\footnotesize{$q^{\sigma(e)}$}};


        \foreach \x in {3,5,7}
        {
          \vertex{2,\x+0.5};
          \vertex{2,\x-0.5};
          \vertex{4,\x};
          \draw (-2,5) -- (0,\x);
        }

        \draw (0,7) -- (2, 7.5) -- (4,7) -- (2,4.5) -- (0,5);
        \draw (0,5) -- (2, 5.5) -- (4,5) -- (2,6.5) -- (0,7);
        \draw (0,3) -- (2, 3+0.5) -- (4,3) -- (2, 3-0.5) -- (0,3);

        \vertex{0,3};
        \vertex{0,5};
        \vertex{0,7};
        \vertex{-2,5};

        \vertex{-2,4};
        \vertex{-2.4,4};
        \draw (-2,5) -- (-2,4);
        \draw (-2,5) -- (-2.4,4);

        \node at (-2.5,3.5) {\scriptsize{$s$} pendent vertices};

      \end{tikzpicture}}
    \caption{A connected component $Z$ of $F_{\{H\}} \gm S_P$ that is not a copy of $K$, with $T_{e} = \{(1,1), (2,2), (3,3)\}$, $T_{\sigma(e)} = \{(1,1), (2,3), (3,2)\}$, and $s=2$.}

    \label{reallocatedbissstretchnewbis}
  \end{figure}
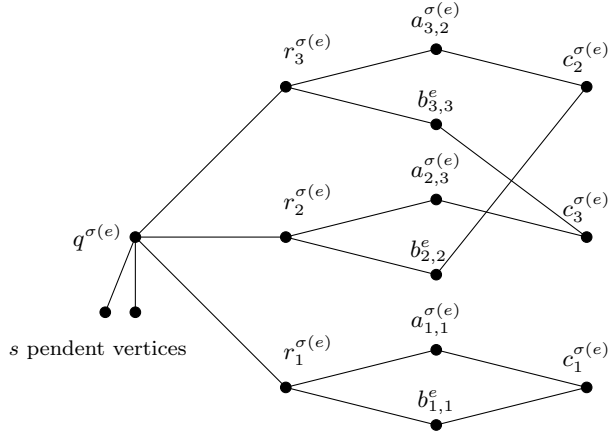
\end{proof}

We are now ready to prove Theorem~\ref{hardminor} and Theorem~\ref{hardtminor}.

\begin{proof}[Proof of Theorem~\ref{hardminor} and Theorem~\ref{hardtminor}]
  Let $H$ be a graph in $\Qcal$.
  If $H$ is a star with at least four leaves, then by Lemma~\ref{twocut} \textsc{$\{H\}$-M-Deletion} cannot be solved in time $2^{\smallo(\tw \log \tw)}\cdot n^{\Ocal(1)}$ unless the \ETH fails.
  In the following we assume that $H$ is not a star. This permits us to proceed with the proofs of both theorems in a unified way.

  If $H$ contains at least one block with at least five vertices, then such a block would have at least five edges as well (by definition of a block), hence by Theorem~\ref{hardgene}, the theorems hold. We can now assume that $H$ does not contain any block with at least five vertices. Therefore,
  every block of $H$ is either an edge, a $C_{3}$, or a $C_{4}$.

  If $H$ contains at least three cut vertices that do not belong to the same block, then Lemma~\ref{threecut} can be applied.
  We can now assume that $H$ contains at most two cut vertices.

  Assume now that $H$ contains exactly two cut vertices and let $B$ the block that contains both of them. If $B$ is not an edge, then Lemma~\ref{cycletwocut} can be applied. Otherwise, we distinguish cases depending on the shape of the two connected components of $H$ after removing the only edge of $B$. If both connected components contain a cycle then Lemma~\ref{cyclestartwocut} can be applied, if only one of them contains a cycle then Lemma~\ref{commander} can be applied, and if none of them contains a cycle then, as $H$ is not a minor of the \banner,  Lemma~\ref{twocut} can be applied.

  Assume now that $H$ contains exactly one cut vertex. As $H$ is not a star, then either $H$ contains at least two cycles, and so Lemma~\ref{butternutbis} can be applied, or $H$ contains exactly one cycle (and is not a minor of the \banner) and therefore Lemma~\ref{crickbisminor} or Lemma~\ref{crickbistm} can be applied. The theorems follow.
\end{proof}

\section{Conclusions and further research}
\label{upheaval}

We provided lower bounds for  \textsc{$\Fcal$-M-Deletion} and  \textsc{$\Fcal$-TM-Deletion} parameterized by the treewidth of the input graph, several of them being tight. In particular, the results of this article together with those of~\cite{monster1,monster2,BasteST20-SODA,SODA-arXiv} settle completely the complexity of \textsc{$\{H\}$-M-Deletion} when $H$ is connected.

Concerning the topological minor version, in order to establish a dichotomy for \textsc{$\{H\}$-TM-Deletion} when $H$ is planar and connected, it remains to obtain algorithms in time $2^{\Ocal(\tw  \cdot \log \tw)}\polyn$ for the graphs $H$ with maximum degree at least four, like the \gem or the \dart (see Figure~\ref{shifting}), as for those graphs the algorithm in time $2^{\Ocal(\tw  \cdot \log \tw)}\polyn$ given in~\cite{monster1} cannot be applied.

It is easy to check that the lower bounds presented in this article also hold for {\sl treedepth} (as it is the case in~\cite{Pili15}) which is a parameter more restrictive than treewidth~\cite{BougeretS16}.  Also, it is worth mentioning that \textsc{$\Fcal$-M-Deletion} and \textsc{$\Fcal$-TM-Deletion} are unlikely to admit polynomial kernels parameterized by treewidth for essentially any collection $\Fcal$, by using the framework introduced by Bodlaender et al.~\cite{BodlaenderDFH09} (see~\cite{BougeretS16} for an explicit proof for any problem satisfying a generic condition).

Finally, let us mention that Bonnet et al.~\cite{BonnetBKM-IPEC17} recently studied generalized feedback vertex set problems parameterized by treewidth, and showed that excluding $C_4$ plays a fundamental role in the existence of single-exponential algorithms. This is related to our dichotomy for cycles illustrated in Figure~\ref{shifting} (which we proved independently in~\cite{BasteST17} building on the work of Pilipczuk~\cite{Pili15}), namely that \textsc{$\{C_i\}$-Deletion} can be solved in single-exponential time if and only if $i \leq 4$.

\vspace{.3cm}

\noindent \textbf{Acknowledgements}. We would like to thank the referees of the two conference versions containing some of the results of this article for helpful remarks that improved the presentation of the manuscript, and  \href{http://www.lamsade.dauphine.fr/~bonnet/}{Édouard Bonnet}, \href{http://www.lamsade.dauphine.fr/~kim/}{Eun Jung Kim},  and \href{https://mat-web.upc.edu/people/juan.jose.rue/}{Juanjo Rué} for insightful discussions on the topic of this paper.

\bibliographystyle{abbrv}
\bibliography{Biblio-Fdeletion3}

\end{document}